\documentclass{amsart}

\usepackage{amssymb}
\usepackage{amsmath}
\usepackage[all]{xy}

\usepackage{amsfonts}

\usepackage{amscd}
\usepackage{amsthm}
\usepackage{latexsym}
\usepackage{amsbsy}
\usepackage{graphicx}
\usepackage{subcaption}
\usepackage{textcomp}

\usepackage{bbm}
\usepackage{enumitem}
\usepackage[utf8]{inputenc}
\usepackage[english]{babel}
\usepackage{hyperref}

\usepackage{mathrsfs}

\usepackage{colortbl}
\def\red{\color{red}}
\def\black{\color{black}}
\def\blue{\color{blue}}

\newcommand{\car}{CAR(\mathcal{H})}
\newcommand{\ccr}{CCR(\mathcal{H})}
\newcommand{\fermi}{\mathcal{F_-(H)}}
\newcommand{\bose}{\mathcal{F_+(H)}}
\newcommand{\id}{\mathbbm{1}}

\newcommand{\isEquivTo}[1]{\underset{#1}{\sim}}

\DeclareMathOperator{\tr}{Tr}
\DeclareMathOperator{\Res}{Res}
\DeclareMathOperator{\re}{\mathfrak{R}}

\theoremstyle{plain}

\newtheorem{theorem}{Theorem}[section]

\newtheorem{lemma}[theorem]{Lemma}
\newtheorem{prop}[theorem]{Proposition}

\newtheorem{definition}[theorem]{Definition}

\def\A{\mathcal{A}}

\def\H{\mathcal{H}}

\newtheorem{remark}{Remark}[section]

\makeatletter
\def\section{\@startsection{section}{1}{\z@}{-3.25ex plus -1ex minus
		-.2ex}{1.5ex plus .2ex}{\normalfont\bfseries}}
\def\subsection{\@startsection{subsection}{1}{\z@}{-3.25ex plus -1ex
		minus -.2ex}{1.5ex plus .2ex}{\normalfont\bfseries\itshape}}
\makeatother

\title{Second Quantization and the Spectral Action}
\author{Rui Dong}
\address{Institute for Mathematics, Astrophysics and Particle Physics, Radboud
University Nijmegen, Heyendaalseweg 135, 6525 AJ Nijmegen, The Netherlands.}
\email{rui.dong@math.ru.nl, waltervs@math.ru.nl}
\author{Masoud Khalkhali}
\address{Department of Mathematics,  University of Western Ontario \\ London, Ontario, Canada}
\email{masoud@uwo.ca}

\author{Walter D. van Suijlekom}

\date{\today}

\begin{document}

\begin{abstract}
  We consider both the bosonic and fermionic second quantization of spectral triples  in the presence of a chemical potential. We show that the von Neumann entropy and the average energy of the Gibbs state defined by the bosonic and fermionic grand partition function can be expressed as spectral actions. It turns out that all spectral action coefficients can be given in terms of the modified Bessel functions. In the fermionic case, we show that the spectral coefficients for  the von Neumann entropy, in the limit when the chemical potential $\mu$ approaches $0,$ can be expressed in terms of the Riemann zeta function. This  recovers a result of Chamseddine-Connes-van Suijlekom. 
\end{abstract}

\maketitle

\tableofcontents

\section{Introduction}
The spectral action principle of Connes and Chamseddine was originally developed mainly to give a conceptual and geometric formulation of the standard model of particle physics \cite{Chamseddine:1996zu}. The spectral action can be defined for spectral triples $(A, \mathcal{H}, D)$, even when the algebra $A$ is not commutative. An interesting feature here is the additivity of  the spectral action with respect to the direct sum of spectral triples. Conversely, one can wonder whether any such additive functional on spectral triples is obtained via a spectral action.

In a recent paper \cite{CCS18},  Chamseddine, Connes and van Suijlekom have shown that the von Neumann entropy of the Gibbs state naturally defined by a fermionic second quantization of a spectral triple is in fact a spectral action and they find a universal function that defines the spectral action. 

In this paper we extend this result by incorporating chemical potentials, and by considering both fermionic and bosonic second quantization. In fact we show that the von Neumann entropy and 
the average energy 
of the thermal equilibrium state defined by the bosonic or fermionic grand partition function, with a given chemical potential,  can be expressed as spectral actions. We show that all spectral action coefficients can be expressed in terms of the modified Bessel functions of the second kind. 
In the fermionic case,
we  show that the spectral  action coefficients for  the von Neumann entropy, in the limit  when the chemical potential $\mu$ approaches $0,$ can be expressed in terms of  the Riemann zeta  function. 
This  recovers the recent result of \cite{CCS18}.

It should be noted that without the use of chemical potentials, the natural spectral function for the von Neumann entropy in the bosonic case  is singular at $t=0$, and in fact the corresponding functional is not spectral.

We start in section \ref{sect_2nd-quant} by recalling some of the main concepts and results from the theory of second quantization. Our main results are presented in Sections \ref{sect_fermi} (fermionic) and \ref{sect_bose} (bosonic). 

In searching for a suitable expression of spectral action coefficients in all the cases studied in this paper, we were naturally led to the class of modified Bessel functions of the second kind. Appendix \ref{sect_bessel} recalls some basic properties of these functions. In Appendix \ref{sect_heat-exp} we collect some technical results on the heat expansions used throughout this paper. Finally in Appendix 
 \ref{Appendix_sa} we recall some of the basic definitions related to the spectral action principle. 


\tableofcontents

\section{Basics of second quantization}
\label{sect_2nd-quant}
In this section, mainly to fix our notation and terminology, we shall recall some basic definitions and facts from the theory of second quantization in  quantum statistical mechanics.  We shall largely follow \cite[Section 5.2]{BR97}.
\subsection{Fock space and second quantization}

Let $\mathcal{H}$ be a (complex) Hilbert space. 
For any $n>0$ we denote by $\mathcal{H}^{ n}=\mathcal{H}\otimes \mathcal{H}\otimes \dots \otimes\mathcal{H}$ the $n$-fold tensor product of $\mathcal{H}$ with itself, and write $\mathcal{H}^0=\mathbb{C}.$
The {\em Fock space} $\mathcal{F(H)}$ is  the completion of the pre-Hilbert space
$\bigoplus\limits_{n\geq 0}\mathcal{H}^n.$

Define the projection operators $P_\pm$ on $\mathcal{H}^n$ by
\begin{equation*}
\begin{aligned}
&P_+\left(f_1\otimes f_2\otimes \dots \otimes f_n\right)=\left(n!\right)^{-1}\sum_{\pi\in S_n}f_{\pi(1)}\otimes f_{\pi(2)}\otimes \dots \otimes f_{\pi(n)},\\
&P_-\left(f_1\otimes f_2\otimes \dots \otimes f_n\right)=\left(n!\right)^{-1}\sum_{\pi\in S_n}(-1)^{|\pi|}f_{\pi(1)}\otimes f_{\pi(2)}\otimes \dots \otimes f_{\pi(n)},
\end{aligned}
\end{equation*}
for all $f_1,...,f_n\in \mathcal{H}$.
Since $P_\pm$ are bounded operators with norm $1$ on $\mathcal{H}^n$,
they  can be extended by continuity to bounded projection operators  on the Fock space $\mathcal{F(H)}$.

The {\em bosonic  Fock space} $\bose$ and the {\em fermionic Fock space} $\fermi$ are then defined by
\begin{equation*}
\mathcal{F}_\pm(\mathcal{H})=P_\pm(\mathcal{F(H)}).
\end{equation*}
The corresponding $n$-particle subspaces $\mathcal{H}^n_\pm$ are defined by $\mathcal{H}^n_\pm=P_\pm\mathcal{H}^n$.

\medskip

The structure of the Fock space allows us to amplify a linear operator  on $\mathcal{H}$ to the whole bosonic/fermionic Fock spaces $\mathcal{F}_\pm(\mathcal{H})$. 
In fact, this is a functorial procedure and is commonly referred to as {\em second quantization}.

Let $H$ be a self-adjoint operator on $\mathcal{H}$ with domain $D(\mathcal{H})$.  We define $H_n$ on $\mathcal{H}_\pm^n$ by 
\begin{equation*}
H_n\left(P_\pm\left(f_1\otimes\dots\otimes f_n\right)\right)=
\begin{aligned}
&P_\pm\left(\sum_{i=1}^n f_1\otimes f_2\otimes \dots \otimes Hf_i\otimes \dots \otimes f_n\right) 
\end{aligned}
\end{equation*}
for all $f_i\in D(H)$ and $n>0$, while we set $H_0=0$. The direct sum of the $H_n$ is essentially self-adjoint, and the self-adjoint closure of this direct sum operator is called  the {\em second quantization of the self-adjoint operator} $H$  and it is denoted by $d\Gamma(H)$.  

In particular, let $H=\id$ be the identity operator. Then we have 
\begin{equation*}
d\Gamma(\id)=N,
\end{equation*}
where $N$ is the  number operator on $\mathcal{F}_\pm(\mathcal{H})$, whose domain is defined by
\begin{equation*}
D(N)=\left\{\psi=\{\psi^{(n)}\}_{n\geq 0 }; \,  \sum_{n\geq 0}n^2||\psi^{(n)}||^2< \infty\right\},
\end{equation*}
and for any $\psi\in D(N)$
\begin{equation*}
N\psi = \{n\psi^{(n)}\}_{n\geq 0}.
\end{equation*}

For a unitary operator $U$ on $\mathcal{H}$, first we define $U_n$ on $\mathcal{H}^n_\pm$ by
\begin{equation*}
U_n\left(P_\pm\left(f_1\otimes f_2\otimes \dots \otimes f_n\right)\right)=
\begin{aligned}
&P_\pm\left(Uf_1\otimes Uf_2\otimes \dots\otimes Uf_n\right) 
\end{aligned}
\end{equation*}
and $U_0=\id$, and then extend it to the whole Fock space. We denote this extension  by $\Gamma (U)$, 
called  the {\em second quantization of the unitary operator} $U$,    
\begin{equation*}
\Gamma(U)=\bigoplus\limits_{n\geq 0}U_n
.
\end{equation*}

It is worth noticing that here $\Gamma(U)$ is also a unitary operator on $\mathcal{F}_\pm(\mathcal{H})$.
Also, if $U_t=e^{itH}$ is a strongly continuous one-parameter unitary group acting on $\mathcal{H}$, 
then
\begin{equation*}
\Gamma(U_t)=e^{itd\Gamma(H)}
\end{equation*}
on the Fock spaces $\mathcal{F}_\pm(\mathcal{H})$.

\subsubsection{Hamiltonians and time evolution}
If $H$ is a self-adjoint Hamiltonian operator on the one-particle Hilbert space $\mathcal{H}$,
then the dynamics of the ideal Bose gas and the ideal Fermi gas are described by the Schr\"{o}dinger equation
\begin{equation*}
i\hbar \frac{d\psi_t}{dt}=d\Gamma(H)\psi_t
\end{equation*}
on $\bose$ and $\fermi$,
separately.
We choose the units so that $\hbar=1$.
The solution of the Schr\"{o}dinger equation gives us the evolution
\begin{equation*}
\psi\in \mathcal{F}_\pm(\mathcal{H})\mapsto \psi_t=e^{-itd\Gamma(H)}\psi=\Gamma(e^{-itH})\psi,
\end{equation*}
and the evolution of a bounded observable $A\in \mathcal{B}(\mathcal{F}_\pm(\mathcal{H}))$ is given by conjugation as
\begin{equation*}
A\in \mathcal{B}(\mathcal{F}_\pm(\mathcal{H}))\mapsto \tau_t(A)=\Gamma(e^{itH})A\Gamma(e^{-itH}).
\end{equation*}

Next we shall introduce the Gibbs grand canonical equilibrium state $\omega$ of a  particle system at inverse temperature $\beta\in \mathbb{R}$,  and with  chemical potential $\mu\in \mathbb{R}$. 
Let
$$K_\mu=d\Gamma(H-\mu \id)=d\Gamma(H)-\mu N$$
be the so-called {\em modified Hamiltonian}. Then the Gibbs state $\phi$ is defined by 
\begin{equation*}
\phi(A)=\frac{\textrm{Tr}\left(e^{-\beta K_\mu}A\right)}{\textrm{Tr}\left(e^{-\beta K_\mu}\right)}, \qquad \qquad A\in  \mathcal{B}(\mathcal{F}_\pm(\mathcal{H})).
\end{equation*}
Here we assumed that the operator $e^{-\beta K_\mu}$ is a trace-class operator. The {\em density operator} is defined as
$$
\rho = \frac{e^{-\beta d\Gamma(H -\mu \id)}}{\textrm{\emph{Tr}}(e^{-\beta d\Gamma(H -\mu \id})}
$$
so that we have $\phi(A) = \tr \rho A$. 


For later use, we record the following result.
\begin{lemma}\label{additive_lemma}
  We have   $e^{-d\Gamma\left(H_1\oplus H_2\right)}=e^{-d\Gamma\left(H_1\right)}\otimes e^{-d\Gamma\left(H_2\right)}$, as operators on $\mathcal{F}_\pm(\mathcal{H}_1 \oplus \mathcal{H}_2) \cong \mathcal{F}_\pm(\mathcal{H}_1 ) \otimes \mathcal{F}_\pm(\mathcal{H}_2)$. Moreover,
when $e^{-d\Gamma(H_i)}$ are trace-class operators for $i=1,2$, 
we have for the density operators that $\rho_{12} = \rho_1\otimes\rho_2.$
\end{lemma}


\subsection{CAR and CCR algebras}\label{subsection_CAR_CCR}
Both of the CAR and  CCR algebras are constructed with the help of creation and annihilation  operators, whose definition we now recall. 

Let $\mathcal{H}$ be a complex Hilbert space. 
For each $f\in \mathcal{H}$,
we define the {\em annihilation operator} $a(f)$,  and the {\em creation operator} $a^*(f)$ acting on the Fock space $\mathcal{F(H)}$ by setting $a(f)\psi^{(0)}=0$, $a^*(f)\psi^{(0)}=f \psi^{(0)}$ and $\psi^{(0)} \in \mathcal H_\pm^0$ and on $\mathcal H_\pm^n$ by setting
\begin{equation*}
\begin{aligned}
&a(f)(f_1\otimes f_2\otimes \cdots \otimes f_n)=\sqrt{n}\left(f,f_1\right)f_2\otimes f_3\otimes \cdots \otimes f_n,\\
&a^*(f)(f_1\otimes f_2\otimes \cdots \otimes f_n)=\sqrt{n+1}f\otimes f_2\otimes f_3\otimes \cdots \otimes f_n.
\end{aligned}
\end{equation*}
for all  $ f \in \mathcal{H}$.
One can see that the maps $f\mapsto a(f)$ are anti-linear while the maps $f\mapsto a^*(f)$ are linear. Also, one can show that $a(f)$ and $a^*(f)$ have well-defined extensions to $D(N^{1/2})$, the domain of the operator $N^{1/2}$.
Moreover, we have that $a^*(f)$ is the formal adjoint of $a(f)$;  namely, for any $\phi, \psi\in D(N^{1/2})$,
one has
\begin{equation*}
\left(a^*(f)\phi, \psi\right)=\left(\phi, a(f)\psi\right).
\end{equation*}

We can then define the annihilation operators $a_\pm(f)$ and the creation operators $a^*_\pm(f)$ on the fermionic/bosonic Fock spaces  $\mathcal{F}_\pm(\mathcal{H})$ by
\begin{equation*}
a_\pm(f)=P_\pm a(f)P_\pm, \qquad a^*_\pm (f)=P_\pm a^* (f)P_\pm.
\end{equation*}
Moreover, since the annihilation operator $a(f)$ keeps the subspaces $\mathcal{F}_\pm(\mathcal{H})$ invariant, we have
\begin{equation*}
a_\pm(f)=a(f)P_\pm, \qquad a^*_\pm(f)=P_\pm a^*(f).
\end{equation*}
One computes straightforwardly that on the fermionic  Fock space $\fermi$,
\begin{equation*}\label{car}
\{a_-(f),a_-(g)\}=\{a^*_-(f),a^*_-(g)\}=0,
\qquad \{a_-(f),a^*_-(g)\}=(f,g)\id,
\end{equation*}
and on the bosonic Fock space $\bose$,
\begin{equation*}\label{ccr}
[a_+(f),a_+(g)]=[a^*_+(f),a^*_+(g)]=0,
\qquad [a_+(f),a^*_+(g)]=(f,g)\id.
\end{equation*}
The first relations are called the canonical anti-commutation relations (CAR), and the second relations are called the canonical commutation relations (CCR).

Roughly speaking, the CAR algebra is the algebra generated by the annihilation operators $a_-(f)$ and creation operators $a_-^*(f)$. In fact, we have the following proposition \cite[Prop. 5.2.2]{BR97}:
\begin{prop}
Let $\mathcal{H}$ be a complex Hilbert space, $\fermi$ be the fermionic Fock space, and $a_-(f)$ and $a^*_-(g)$ the corresponding annihilation and creation operators on $\fermi$.  

(1)  For all $f\in \mathcal{H}$, we have 
\begin{equation*}
||a_-(f)||=||f||=||a^*_-(f)||.
\end{equation*}
Therefore both $a_-(f)$ and $a_-^*(g)$ have bounded extensions on $\fermi$.

(2)  Taking $\Omega=(1,0,0,\cdots)$, 
called the vacuum vector, and $\{f_\alpha\}$ an orthonormal basis of $\mathcal{H}$, then
\begin{equation*}
\psi(f_{\alpha_1},...,f_{\alpha_n}):=a_-^*\left(f_{\alpha_1}\right)\cdots a^*_-\left(f_{\alpha_n}\right)\Omega
\end{equation*}
is an orthonormal basis of $\fermi$,
when $\{f_{\alpha_1},...,f_{\alpha_n}\}$ runs over all the finite subsets of the orthonormal basis $\{f_{\alpha}\}.$

(3) The set of bounded operators 
$\{a_-(f), a^*_-(g);f,g\in \mathcal{H}\}$ is irreducible on $\fermi$.
\end{prop}

\begin{definition}
The $C^*$-subalgebra of $\mathcal{B}(\fermi)$ generated by $a_-(f)$, $a_-^*(g)$ and $\id$ is called the {\em CAR algebra}, and will be denoted by $\car$.
\end{definition}

Although the CCR rules look very similar to the CAR rules, however, one can not simply mimic the previous definition of CAR algebras to deduce the definition of CCR algebras. The reason is that the annihilation operators $a_+(f)$ and the creation operators $a_+^*(g)$ are not bounded operators on $\bose$.

First we introduce the set of operators $\{\Phi(f),f\in \mathcal{H}\}$ by
$$\Phi(f)=\frac{a_+(f)+a_+^*(f)}{\sqrt{2}}.$$
Using that the map $f\mapsto a_+(f)$ is anti-linear, and $f\mapsto a_+^*(f)$ is linear, 
we can invert: 
$$a_+(f)=\frac{\Phi(f)+i\Phi(if)}{\sqrt 2}, \quad a_+^*(f)=\frac{\Phi(f)-i\Phi(if)}{\sqrt 2}. $$
Thus it suffices to examine the set of operators $\{\Phi(f),f\in \mathcal{H}\}$.

Let $F_+(\mathcal{H})=P_+\left(\bigoplus_{n\geq 0}\mathcal{H}^n\right)\subseteq \bose$, {\it i.e.} $F_+(\mathcal{H})$ consists of the sequences $\psi=\{\psi^{(n)}\}_{n\geq 0}$ which have only a finite number of nonvanishing components.

It turns out that for each $f\in \mathcal{H}$, $\Phi(f)$ is essentially self-adjoint on $F_+(\mathcal{H})$, so that $\Phi(f)$ can be extended to a self-adjoint operator; we still use the notation $\Phi(f)$ for this closure. 
We have the following result \cite[Prop. 5.2.4]{BR97}:
\begin{prop}
For each $f\in \mathcal{H}$ define a unitary operator 
$W(f)=\exp{(i\Phi(f))}$ on $\bose$.
Let $\ccr$ denote the $C^*$-algebra generated by $\{W(f), f\in \mathcal{H}\}$.
It follows that

(1) For any $f,g\in \mathcal{H}$, $W(f)D(\Phi(g))=D(\Phi(g))$, and
$$W(f)\Phi(g)W(f)^*=\Phi(g)-\emph{Im}(f,g)\id.$$

(2) For each pair $f,g\in \mathcal{H}$
$$W(f)W(g)=e^{-i\textrm{\emph{Im}}(f,g)/2}W(f+g).$$

(3) $W(-f)=W(f)^*$.

(4) For each $f\in \mathcal{H}\backslash \{0\}$
$$||W(f)-\id||=2,$$
and $W(0)=\id.$

(5) The set $\{W(f);f\in \mathcal{H}\}$ is irreducible on $\bose$, and $\ccr$ is a simple algebra.

(6) If $||f_\alpha - f||\rightarrow{0}$, then
$$||\left(W(f_\alpha)-W(f)\right)\psi||\rightarrow 0$$
for all $\psi\in \bose.$
The operators $W(f)$ are called Weyl operators, and the algebra $\ccr$ is called the CCR algebra of $\mathcal{H}$. 

\end{prop}

\subsubsection{Gibbs states, entropy and energy}
As before, let $K_\mu$ denote the modified Hamiltonian operator
\begin{equation*}
K_\mu= d\Gamma\left(H-\mu \id\right).
\end{equation*}

In the fermionic case,
we can define the Gibbs state $\phi_f$ over the CAR algebra by
\begin{equation*}
\phi_f(A)=\frac{\textrm{Tr}\left(e^{-\beta K_\mu}A\right)}{\textrm{Tr}\left(e^{-\beta K_\mu}\right)},\qquad( A\in \car).
\end{equation*}
Here we have tacitly assumed that the operator $e^{-\beta K_\mu}$ is a trace-class operator on $\fermi$, but this in fact happens if and only if $e^{-\beta H}$ is of trace-class on the one-particle Hilbert space $\mathcal H$ \cite[Prop. 5.2.22]{BR97}.





In the bosonic case, 
we can define the Gibbs state $\phi_b$ over the CCR algebra $\ccr$ by
\begin{equation*}
\phi_b(A)=\frac{\textrm{Tr}\left(e^{-\beta K_\mu}A\right)}{\textrm{Tr}\left(e^{-\beta K_\mu}\right)}, \qquad ( A\in \ccr).
\end{equation*}
Similarly as above, it is assumed that the operator $e^{-\beta K_\mu}$ is trace-class on $\bose$, which again is equivalent to the assumption that $e^{-\beta H}$ is trace-class on the one-particle Hilbert space $\mathcal{H}$ \cite[Prop. 5.2.27]{BR97}

\medskip


Under the assumption that $\exp(-\beta H)$ is a trace-class operator on $\mathcal H$, we let $\rho$ denote the corresponding density operator:
$$
\rho = \frac{e^{-\beta K_\mu}}{\textrm{Tr}\left(e^{-\beta K_\mu}\right)}
$$
We will also write $\rho_f$ and $\rho_b$ to stress whether we are dealing with the fermionic or the bosonic case.

The {\em von Neumann entropy} is then defined to be $$\mathcal{S}(\rho):=-\textrm{Tr}(\rho\log \rho),$$
while the {\em average energy} $ \langle K_\mu \rangle_\beta$ is given by 
\begin{equation}\label{eq_def_av_eng}
\langle K_\mu \rangle_\beta=\textrm{Tr}(\rho K_\mu)=-\frac{\partial}{\partial \beta}\left(\log Z\right).
\end{equation}

\subsection{Spectral triples and second quantization}
In this paper we are interested in the above concepts of Fock space, entropy, average energy, {\em et cetera} such as they derive from spectral triples in noncommutative geometry. We recall the definition \cite{C94}:
\begin{definition}
  A {\em spectral triple} is a triple $(\mathcal A,\H,D)$ where $\mathcal A$ is a $*$-algebra of bounded operators acting on a Hilbert space $\H$, and $D$ is a self-adjoint operator in $\H$ with compact resolvent and such that $[D,a]$ is a bounded operator for all $a \in \A$.
\end{definition}

Starting with a spectral triple, we can construct the bosonic and fermionic Fock spaces $\bose$ and  $\fermi$, respectively. 
In the fermionic case,
if $\exp(-\beta|D|)$ is trace-class on $\mathcal{H}$ then, as before, %
the operator $\exp(-\beta d\Gamma |D|)$ is trace-class on $\fermi$,
We denote the corresponding density matrix by $\rho_D$ and Gibbs state by $\phi_D$, The von Neumann entropy $\mathcal{S}(\rho_D)$ and average energy $\langle d\Gamma|D|\rangle_\beta$ are thus defined.
Note that both of these two quantities are additive in the sense that if $D$ has a direct sum decomposition $D=S\oplus T$,
then 
$$\mathcal{S}(\rho_D)=\mathcal{S}(\rho_S)+\mathcal{S}(\rho_T),\quad \textrm{and}\quad \langle d\Gamma|D|\rangle_\beta=\langle d\Gamma|S|\rangle_\beta+\langle d\Gamma|T|\rangle_\beta.$$

Arguing as in \cite{CCS18} we thus obtain two `spectral actions':
$$D\mapsto \mathcal{S}(\rho_D),\quad \textrm{and}\quad D\mapsto \langle d\Gamma|D|\rangle_\beta$$
associated naturally to the spectral triple $(\mathcal{A},\mathcal{H},D)$.

\medskip

Even though for the motivating example of a spectral triple given by a Riemannian spin$^c$ manifold it is natural to consider vectors in the Hilbert space as fermions, in general there is no reason to exclude the possibility that they are bosonic instead. We include this possibility here as well, thus replacing the fermionic Fock space by the bosonic Fock space. 

Then, when the second quantization operator $\exp(-\beta d\Gamma D^2)$ is trace-class on $\bose$ we can obtain another two spectral actions given by the von Neumann entropy and average energy. 
However, in general the operator $\exp(-\beta d\Gamma D^2)$ is not trace-class, which is another reason for including the chemical potential $\mu<0$.

In Section \ref{sect_fermi} we shall first discuss the von Neumann entropy and average energy in fermionic Fock space, starting with following one-particle Hamiltonians:
$$
H_{f,\mu}:=\sqrt{D^2+\mu^2\id},\quad \textrm{and}\quad H_{f,\mu}':=|D|-\mu\id.
$$
This generalizes the case considered in \cite{CCS18} without chemical potential.
We will see that even though the operator $H_{f,\mu}'$ is more natural to consider from a physical point of view (as it is the Hamiltonian based on the Dirac sea in the presence of a chemical potential), the operator $H_{f,\mu}$ gives rise to interesting mathematical structure hidden in the entropy and average energy.

Then, in Section \ref{sect_bose} 
we shall analyze the same two concepts but now in bosonic Fock space and for the following one-particle Hamiltonians:
$$
H_{b,\mu}:=\sqrt{D^2+\mu^2\id},\quad \textrm{and} \quad H_{b,\mu}':=D^2-\mu\id.
$$
Again the latter is the one-particle Hamiltonian considered in physics.

\section{Fermionic second quantization}\label{sect_fermi}
Let $D:\mathcal{H}\to\mathcal{H}$ be a self adjoint operator with compact resolvent.
Let $\beta>0$.
If $e^{-\beta|D|}$ is a trace class operator,
then for any $\mu<0$ the one-particle Hamiltonians $H_{f,\mu},H_{f,\mu}'$ as defined above,
give rise to two trace class operators $e^{-\beta H_{f,\mu}}$ and $e^{-\beta H_{f,\mu}'}$. In fact, since 
$$\sqrt{D^2+\mu^2\id}-|D|=\mu^2\left(\sqrt{D^2+\mu^2\id}+|D|\right)^{-1}$$
is a bounded operator, we obtain the following equalities:
\begin{align*}
\tr\left(e^{-\beta H_{f,\mu}}\right)
&=
\tr\left(e^{-\beta|D|}e^{-\beta\mu^2(\sqrt{D^2+\mu^2\id}+|D|)^{-1}}\right),\\
\tr\left(e^{-\beta H_{f,\mu}'}\right)
&=
e^{-\beta\mu}\tr\left(e^{-\beta|D|}\right).
\end{align*}
ensuring the finiteness of the trace.

Let $K_{f,\mu} = d\Gamma H_{f,\mu}$  and $K_{f,\mu'} = d\Gamma H_{f,\mu}'$
denote the second-quantized operators acting as unbounded operators on the fermionic Fock space $\fermi$. We write $\rho_f, \rho_f'$ for the respective density matrices and $\phi_f, \phi_f'$ for the corresponding states. 
We will now compute their von Neumann entropy and their average energy.

\subsection{The entropy of $\rho_f$}
We start by introducing the following function:
\begin{equation*}
h_{\mu}(x) = \frac{\sqrt{x^2+\mu^2}}{e^{\sqrt{x^2+\mu^2}}+1}+\log\left(1+e^{-\sqrt{x^2+\mu^2}}\right).
\end{equation*}
We may imitate the proof of \cite[Theorem 3.4]{CCS18} to show that the von Neumann entropy of $\rho_f$ is given by 
\begin{equation*}
\mathcal{S}(\rho_f)
=
-\tr\left(\rho_{f} \log{\rho_{f}}\right)
=
\tr(h_{\beta\mu}(\beta D)).
\end{equation*}

We now derive the asymptotic expansion of $\tr(h_{\beta\mu}(\beta D))$ for large temperature (that is, as $\beta\to 0^+$).
Let us first show that the function $h_\mu(\sqrt{x})$ can be expressed as a Laplace transform.
According to Proposition 4.4 in \cite{CCS18},
\begin{equation*}
h_{0}(x)=\int_{0}^{\infty}e^{-tx^2}\Tilde{g}(t)dt,
\end{equation*}
where 
\begin{equation*}
\Tilde{g}(t)=\frac{1}{2t}\sum\limits_{n\in \mathbb{Z}}\left(2\pi^2(2n+1)^2t-1\right)e^{-\pi^2(2n+1)^2t}.
\end{equation*}
Thus
\begin{equation*}
h_\mu(x)=\int_0^\infty e^{-t(x^2+\mu^2)} \Tilde{g}(t)dt=\int_{0}^\infty e^{-tx^2}\Tilde{g}_{\mu}(t)dt,
\end{equation*}
where  $\Tilde{g}_{\mu}(t): =e^{-t\mu^2}\Tilde{g}(t)$.

As a preparation for the asymptotic expansion of the entropy for large temperature we now derive some results expressing the moments of the function $h_\mu(x)$, that is to say, the integrals
$$\int_{0}^\infty h_\mu(x)x^\nu dx.$$
To this end, we first compute the two integrals 
$$\int_{0}^\infty \log \left(1+e^{-\sqrt{x^2 +\mu^2}}\right)x^\nu dx
\quad
\textrm{and}
\quad
\int_{0}^\infty\frac{\sqrt{x^2+\mu^2}}{e^{\sqrt{x^2+\mu^2}}+1}x^\nu dx$$ separately, 
and then sum them up.

\begin{lemma}\label{int1}
We have the following formulae:
\begin{equation}\label{1stderiv}
\int_1^\infty e^{-zx}(x^2-1)^{\nu-\frac{1}{2}}xdx
=
\frac{2^\nu}{\sqrt{\pi}}\Gamma\left(\nu+\frac{1}{2}\right)z^{-\nu}K_{\nu+1}(z)
\end{equation}
and
\begin{equation}\label{2ndderiv}
\int_1^\infty e^{-zx}(x^2-1)^{\nu-\frac{1}{2}}x^2dx
=
\frac{2^\nu}{\sqrt{\pi}}\Gamma\left(\nu+\frac{1}{2}\right)z^{-\nu-1}\big(zK_\nu(z)+(1+2\nu)K_{\nu+1}(z)\big).
\end{equation}
\end{lemma}

\begin{proof}
According to Lemma \ref{intrep}, one has the integral formula:
\begin{equation*}
\int_1^\infty e^{-zx}(x^2-1)^{\nu-\frac{1}{2}}xdx
=
-\frac{2^\nu }{\sqrt{\pi}}\Gamma\left(\nu+\frac{1}{2}\right)z^{-\nu}\left(\frac{\partial}{\partial z}K_{\nu}(z)-\nu K_\nu(z)z^{-1}\right),
\end{equation*}
which combined with \eqref{bess3} and \eqref{bess4} yields Equation \eqref{1stderiv}.

On the other hand,
taking the derivative with respect to $z$ on both sides of \eqref{1stderiv},
we obtain
\begin{equation}\label{2ndderivproof}
\int_1^\infty e^{-zx}(x^2-1)^{\nu-\frac{1}{2}}x^2dx
=
\frac{2^\nu}{\sqrt{\pi}}\Gamma\left(\nu+\frac{1}{2}\right)\left(\nu z^{-\nu-1}K_{\nu+1}(z)-z^{-\nu}\frac{\partial}{\partial z}K_{\nu+1}(z)\right).
\end{equation}
Taking into account Equation \eqref{bess3} again, we obtain \eqref{2ndderiv}.
\end{proof}

Now we can obtain the integrals that we after as follows:

\begin{lemma}\label{1intprop}
When $\nu>-1$, one has the formulae:
\begin{equation}\label{1int}
\begin{aligned}
\quad\int_{0}^\infty \log \left(1+e^{-\sqrt{x^2 +\mu^2}}\right)x^\nu dx
=
|\mu|^{\frac{\nu+2}{2}}\, 2^{\frac{\nu}{2}}\, \frac{1}{\sqrt{\pi}}\, \Gamma\left(\frac{\nu+1}{2}\right)\sum_{n=1}^{\infty}\frac{(-1)^{n+1}}{n^{\frac{\nu}{2}+1}}K_{\frac{\nu}{2}+1}\left(n|\mu|\right),
\end{aligned}
\end{equation}
and
\begin{equation}\label{2int}
\begin{aligned}
\quad \int_0^\infty \frac{\sqrt{x^2+\mu^2}}{e^{\sqrt{x^2+\mu^2}}+1}x^\nu dx
=
\frac{|2\mu|^{\frac{\nu}{2}}}{\sqrt{\pi}}\, \Gamma\left(\frac{\nu+1}{2}\right)\sum_{n=1}^\infty (-1)^{n+1}\left(\frac{|\mu|}{n^{\frac{\nu}{2}}}K_{\frac{\nu}{2}}\left(n|\mu|\right)+\frac{1+\nu}{n^{\frac{\nu}{2}+1}}K_{\frac{\nu}{2}+1}\left(n|\mu|\right)\right).
\end{aligned}
\end{equation}

\end{lemma}

\begin{proof}
We first notice that 
\begin{equation}\label{1stsum}
\int_0^\infty\log\left(1+e^{-\sqrt{x^2+\mu^2}}\right)x^{\nu}dx
=
\sum_{n=1}^\infty(-1)^{n+1}\frac{1}{n}\int_0^\infty e^{-n\sqrt{x^2+\mu^2}}x^{\nu}dx.
\end{equation}
Let $z=\sqrt{\frac{x^2}{\mu^2}+1}$,
substitute $x$ by $z$ we obtain:
\begin{equation*}
\int_0^\infty e^{-n\sqrt{x^2+\mu^2}}x^{\nu}dx
=
|\mu|^{\nu+1}\int_1^\infty e^{-n|\mu|z}(z^2-1)^{\frac{\nu-1}{2}}zdz.
\end{equation*}
Thus using Lemma \ref{int1},
one has
\begin{equation}\label{1stintegral}
\int_0^\infty e^{-n\sqrt{x^2+\mu^2}}x^{\nu}dx
=
\frac{2^{\frac{\nu}{2}}}{\sqrt{\pi}}\Gamma\left(\frac{\nu+1}{2}\right)\left(n|\mu|\right)^{-\frac{\nu}{2}}K_{\frac{\nu}{2}+1}\left(n|\mu|\right).
\end{equation}
Taking \eqref{1stintegral} into \eqref{1stsum}, 
we get formula \eqref{1int}.

On the other hand,
since
\begin{equation*}
\frac{\sqrt{x^2+\mu^2}}{1+e^{\sqrt{x^2+\mu^2}}}=\sum_{n=1}^{\infty}(-1)^{n+1}\sqrt{x^2+\mu^2}e^{-n\sqrt{x^2+\mu^2}},
\end{equation*} 
we obtain the formula:
\begin{equation*}
\int_0^\infty \frac{\sqrt{x^2+\mu^2}}{e^{\sqrt{x^2+\mu^2}}+1}x^\nu dx
=|\mu|^{\nu+1}\sum_{n=1}^\infty(-1)^{n+1}\int_1^\infty e^{-n|\mu|z}\left(z^2-1\right)^{\frac{\nu-1}{2}}z^2dz.
\end{equation*}
Now applying  Lemma \ref{int1} again,
one gets the integral formula \eqref{2int}.
\end{proof}
If we sum up the formulae \eqref{1int} and \eqref{2int},
we can obtain the $\nu-$th moment of the function $h_\mu(x)$:
\begin{lemma}\label{fermimomentprop}
For $\nu>-1$, one has
\begin{equation*}
\begin{aligned}
\int_{0}^\infty h_\mu(x)x^{\nu} dx
= \frac{|\mu|^{\frac{\nu}{2}+2}2^{\frac{\nu}{2}}}{\sqrt{\pi}}\Gamma\left(\frac{\nu+1}{2}\right)\sum_{n=1}^{\infty}(-1)^{n+1}\left(n^{-\frac{\nu}{2}}K_{\frac{\nu}{2}+2}\left(n|\mu|\right)\right).
\end{aligned}
\end{equation*}
\end{lemma}

Let us write 
\begin{equation}\label{spectral_action_coeff}
\gamma_{\mu}(a)
:=
\int_0^\infty t^a\Tilde{g}_{\mu}(t)dt
=
\int_0^\infty t^ae^{-t\mu^2}\Tilde{g}(t)dt.
\end{equation}
It is clear that for a  fixed chemical potential $\mu<0$,  the equation \eqref{spectral_action_coeff} is an entire function with respect to $a\in \mathbb{C}$.
In view of Lemma \ref{fermimomentprop}, we deduce that when the order $a<0$, the coefficient of $t^{a}$ in the heat expansion is
\begin{align}
\gamma_{\mu}(a)
&=
\frac{1}{\Gamma(-a)}\int_0^\infty h_{\mu}(x^{\frac{1}{2}})x^{-a-1}dx\nonumber
=
\frac{2}{\Gamma(-a)}\int_0^\infty h_{\mu}(x)x^{-2a-1}dx\nonumber\\
&=
\frac{1}{\sqrt{\pi}}2^{-a+\frac{1}{2}}|\mu|^{-a+\frac{3}{2}}\sum_{n=1}^\infty (-1)^{n+1}\left(n^{a+\frac{1}{2}}K_{-a+\frac{3}{2}}\left(n|\mu|\right)\right)\label{fermi_coef_fun}, 
\quad a<0.
\end{align}

Now we will show that for any fixed chemical potential $\mu<0$, 
the function \eqref{fermi_coef_fun} is an entire function with respect to $a\in \mathbb{C}$ ,
so that the function \eqref{fermi_coef_fun} can give rise to spectral action coefficients for any  order $a$.


\begin{prop}\label{thm_fermi_coeff}
For any fixed chemical potential $\mu<0$, the function $\eqref{fermi_coef_fun}$ is an entire  function in $a\in \mathbb{C}$. 
Hence we have the formula
\begin{equation}\label{eq_fermi_coeff_gamma}
\gamma_{\mu}(a)
=
\frac{1}{\sqrt{\pi}}2^{-a+\frac{1}{2}}|\mu|^{-a+\frac{3}{2}}\sum_{n=1}^\infty (-1)^{n+1}\left(n^{a+\frac{1}{2}}K_{-a+\frac{3}{2}}\left(n|\mu|\right)\right)
\end{equation}
for all  $a$. 
\end{prop}

\begin{proof}
We only need to show that the series
\begin{equation}\label{uniform_convergent_function}
\sum_{n=1}^\infty (-1)^{n+1}\left(n^{a+\frac{1}{2}}K_{-a+\frac{3}{2}}\left(n|\mu|\right)\right)
\end{equation}
is an entire function in  $a\in \mathbb{C}$.
In fact, using the integral expression for the Bessel  function $K_\nu(z)$ \cite[8.432]{gradshteyn2007},
we have
\begin{equation*}
K_\nu (z)=
\int_0^\infty e^{-z\cosh{t}}\cosh(\nu t)dt, \qquad |\textrm{arg}z|<\frac{\pi}{2} \textrm{ or } \textrm{Re}(z)=0 \textrm{ and }\nu=0.
\end{equation*}
We see that for a fixed $z>0$
the function $K_\nu(z)$ is an entire function with respect to $\nu\in \mathbb{C}$.
Now we need to show that equation \eqref{uniform_convergent_function} is locally uniformly convergent.
In fact, for $|\nu|\leq R$, 
\begin{equation*}
|K_v(z)|\leq \int_0^\infty e^{-z\cosh{t}}\cosh(Rt)dt
=K_R(z).
\end{equation*}
For $|-a+\frac{3}{2}|\leq R$, where $R<\infty$, we have
\begin{equation*}
\Big|\sum_{n=1}^\infty (-1)^{n+1}\left(n^{a+\frac{1}{2}}K_{-a+\frac{3}{2}}\left(n|\mu|\right)\right)\Big|
\leq
\sum_{n=1}^\infty n^{R+2}K_R\left(n|\mu|\right).
\end{equation*}
According to Lemma \ref{lemma_bessel_asymp},
 we have the asymptotic expansion
\begin{equation*}
K_{\nu}(z)\sim \sqrt{\frac{\pi}{2z}}e^{-z} \qquad z\rightarrow \infty,
\end{equation*}
from which it follows that  the series
$\sum_{n=1}^\infty n^{R+2}K_R\left(n|\mu|\right)$ is convergent.
Therefore the series  $\eqref{uniform_convergent_function}$ is locally uniformly convergent, and  the function $\eqref{eq_fermi_coeff_gamma}$ is an entire function.  Since according to Proposition \ref{thm_fermi_coeff}, $\gamma_{\mu}(a)$ is an entire function, the expression $\eqref{eq_fermi_coeff_gamma}$ is valid for all $a$.
\end{proof}

In addition, we can express the spectral action coefficients $\gamma_{\mu}(a)$ in a more concise way via the Poisson summation formula.

\begin{prop}\label{fermi_spectral_action_coeff_2nd_prop}
For any fixed chemical potential $\mu<0$,  we have 
\begin{equation}\label{2nd_eq__fermi_entropy}
\gamma_{\mu}(a)
=
\frac{\Gamma(a)}{2}\sum_{n=-\infty}^\infty\frac{(2a-1)(2n+1)^2\pi^2-\mu^2}{((2n+1)^2\pi^2+\mu^2)^{a+1}}.
\end{equation}
\end{prop}

\begin{proof}
Using  Lemma \ref{lemma_2_order_FT} and  the Poisson summation formula, when $\nu\geq -\frac{1}{2}$, $a>0$, we have
\begin{equation}\label{fermi_entropy_fourier_transf}
\sum_{n=1}^\infty(-1)^n |n|^{\nu+2}K_\nu(a|n|)
=
\frac{1}{2}\frac{\Gamma\left(\nu+\frac{1}{2}\right)(2a)^\nu}{-4\pi\Gamma\left(\frac{1}{2}\right)}\sum_{n=-\infty}^\infty \phi_{\nu,a}^{\prime\prime}(n),
\end{equation}
where $\phi_{\nu,a}(x)=\frac{1}{((2x+1)^2\pi^2+a^2)^{\nu+\frac{1}{2}}}$.
Since we have the equation
\begin{equation*}
K_{-a+\frac{3}{2}}\left(n|\mu|\right)=K_{a-\frac{3}{2}}\left(n|\mu|\right),
\end{equation*}
applying Equation \eqref{fermi_entropy_fourier_transf} to Proposition \ref{thm_fermi_coeff}
we then get Equation \eqref{2nd_eq__fermi_entropy} when $a\geq \frac{3}{2}$.
Now, in Proposition \ref{thm_fermi_coeff} we saw that $\gamma_{\mu}(a)$ is an entire function.
It follows that  the function \eqref{2nd_eq__fermi_entropy} has an analytic extension to the whole complex plane $\mathbb{C}$, and 
therefore equation \eqref{2nd_eq__fermi_entropy} holds for all $a\in \mathbb{C}$.
\end{proof}

Whenever $\mu\in (-\pi,0)$ we can express the function $\gamma_\mu(a)$ in terms of Riemann $\xi$ function as well, similar to \cite[Prop. 4.6]{CCS18}.

\begin{lemma}\label{prop_gamma_mu_xi}
When $\mu\in(-\pi, 0)$,
the function $\gamma_\mu(a)$ can be expressed as
\begin{equation}\label{3rd_eq_fermi_entropy}
\gamma_\mu(a)=\sum_{k=0}^\infty(-1)^k\frac{1-2^{-(2a+2k)}}{\Gamma(k+1)\pi^{a+k}(a+k)}\xi(2a+2k)\mu^{2k}.
\end{equation}
\end{lemma}

\begin{proof}
We will try to modify formula \eqref{2nd_eq__fermi_entropy} to get to \eqref{3rd_eq_fermi_entropy}.
Since $\mu\in(-\pi,0)$,
\begin{equation*}
\begin{aligned}
\frac{1}{((2n+1)^2\pi^2+\mu^2)^a}
&=
\frac{1}{((2n+1)^2\pi^2)^a}\frac{1}{\left(1+\frac{\mu^2}{(2n+1)^2\pi^2}\right)^a}\\
&=
\frac{1}{((2n+1)^2\pi^2)^a}\sum_{k=0}^\infty\frac{\Gamma(-a+1)}{\Gamma(k+1)\Gamma(-a-k+1)}\left(\frac{\mu^2}{(2n+1)^2\pi^2}\right)^k,
\end{aligned}
\end{equation*}
which we may substitute in formula \eqref{2nd_eq__fermi_entropy} to obtain
\begin{equation*}
\gamma_\mu(a)=\sum_{k=0}^\infty\frac{(-1)^k\Gamma(k+a)(2a+2k-1)}{\Gamma(k+1)}\frac{2^{2a+2k}-1}{(2\pi)^{2a+2k}}\zeta(2a+2k)\mu^{2k}.
\end{equation*}
Since 
\begin{equation*}
\xi(s)=\frac{s(s-1)}{2\pi^{s/2}}\Gamma\left(\frac{s}{2}\right)\zeta(s),
\end{equation*}
we obtain the formula \eqref{3rd_eq_fermi_entropy}.
\end{proof}

We shall follow the same notation as in \cite{CCS18} and denote by
\begin{equation*}\label{eq_gamma}
\gamma(a)=\frac{1-2^{-2a}}{a}\pi^{-a}\xi(2a). 
\end{equation*}
We then have a more concise expression of $\gamma_\mu(a)$:
\begin{equation}\label{eq_espresses_gamma_mu}
\gamma_\mu(a)=\sum_{k=0}^\infty(-1)^k\frac{\gamma(a+k)}{k!}\mu^{2k},\quad \mu\in(-\pi,0),
\end{equation}
leading to the following proposition:

\begin{prop}
Under the assumption of a heat trace expansion of the form \eqref{eq_ht},
we write $\psi_l(\beta, \mu):=\sum\limits_{z\in X_l}\beta^{-2z}\gamma_{\beta\mu}(-z)$.

For $\mu<0$ we then have the following expansion for the entropy:
\begin{equation*}
\mathcal{S}(\rho_f)
=
\tr(h_{\beta\mu}(\beta D))
\sim
\sum_{l}^\infty \psi_l(\beta, \mu),\quad \beta\to 0^+.
\end{equation*}
More precisely,
$$
\tr(h_{\beta\mu}(\beta D))
-
\sum_{l=0}^L \psi_l(\beta, \mu)
=
o_0(\beta^{2L}),
$$
and each term $\psi_{l}(\beta, \mu)=o_0(\beta^{r_l})$.
\end{prop}
\noindent Here we refer to Appendix \ref{Appendix_asy_exp} for the notation $r_l$ in the above Proposition.

\black

\subsection{The average energy of $K_{f,\mu}$}\label{section_fermi_av_eng}
Now we shall consider the average energy $\langle K_{f,\mu} \rangle_\beta = \langle d\Gamma H_{f,\mu}\rangle_\beta$.
Let
\begin{equation*}
Z_f
=
\tr(e^{-\beta d\Gamma H_{f,\mu}})
\end{equation*}
be the corresponding canonical partition function.
 Recall  the formula \eqref{eq_def_av_eng} for  average energy. We see that 
 the average energy with respect to the operator $d\Gamma H_{f,\mu}$ is given by
\black
\begin{equation*}
\langle d\Gamma H_{f,\mu}\rangle_\beta
=
-\frac{\partial }{\partial \beta}\log Z_{f}
=
\tr\left(\frac{H_{f,\mu}}{1+e^{\beta H_{f,\mu}}}\right)
\equiv \frac{1}{\beta}\tr(u_{\beta\mu}(\beta D))
,
\end{equation*}
in terms of the function $u_{ \mu}(x):=\frac{\sqrt{x^2+\mu^2}}{1+e^{\sqrt{x^2+\mu^2}}}$. 

\black

We observe that the function $u_{\mu}(x)$ is just the first part of the function $h_\mu(x)$ used to describe the entropy above.
Let us first show that $u_\mu(\sqrt x)$ can be expressed as a Laplace transform.
Indeed, since we have the expansion 
\begin{equation*}
u_\mu(x)
=
\frac{\sqrt{x^2+\mu^2}}{1+e^{\sqrt{x^2+\mu^2}}}
=
\sum_{n=1}^\infty (-1)^{n+1}\sqrt{x^2+\mu^2}e^{-n\sqrt{x^2+\mu^2}},
\end{equation*}
and 
({\em cf.} e.g. \cite[P146, Formula(29)]{bateman1954tables})
\black
\begin{equation*}
\sqrt{x}e^{-n\sqrt{x}}
=
\frac{1}{\sqrt{\pi}}\int_0^\infty t^{-5/2}\left(\frac{n^2}{4}-\frac{t}{2}\right)e^{-\frac{n^2}{4t}}e^{-tx}dt,
\end{equation*}
for a fixed $\mu<0$, we obtain 
\begin{equation*}
u_\mu(x)
=
\sum_{n=1}^\infty(-1)^{n+1}
\frac{1}{\sqrt{\pi}}\int_0^\infty t^{-5/2}\left(\frac{n^2}{4}-\frac{t}{2}\right)e^{-\frac{n^2}{4t}}e^{-t(x^2+\mu^2)}dt. 
\end{equation*}
Using the Fubini theorem we can switch the order of infinite sum and integral, 
so that
\begin{equation*}
u_\mu(x)
=
\frac{1}{\sqrt{\pi}}\int_0^\infty t^{-5/2}\sum_{n=1}^\infty(-1)^{n+1}\left(\frac{n^2}{4}-\frac{t}{2}\right)e^{-\frac{n^2}{4t}}e^{-t(x^2+\mu^2)}dt, \qquad x\geq 0. 
\end{equation*}
In order to obtain the sought-for Laplace transform, we introduce
\begin{equation*}
r_{\mu}(t)
:=
\frac{1}{\sqrt{\pi}} t^{-5/2}\sum_{n=1}^\infty(-1)^{n+1}\left(\frac{n^2}{4}-\frac{t}{2}\right)e^{-\frac{n^2}{4t}}e^{-t\mu^2}. 
\end{equation*}
Then we obtain the following expression for the Laplace transform of $r_{\mu}$:
\begin{equation*}
u_\mu(\sqrt x)
=
\int_0^\infty r_{\mu}(t)e^{-tx}dt, \qquad \mu<0,  \quad  x\geq 0.
\end{equation*}

\begin{figure}[h]

\begin{subfigure}{0.49\textwidth}
\includegraphics[width=0.9\linewidth, height=5cm]{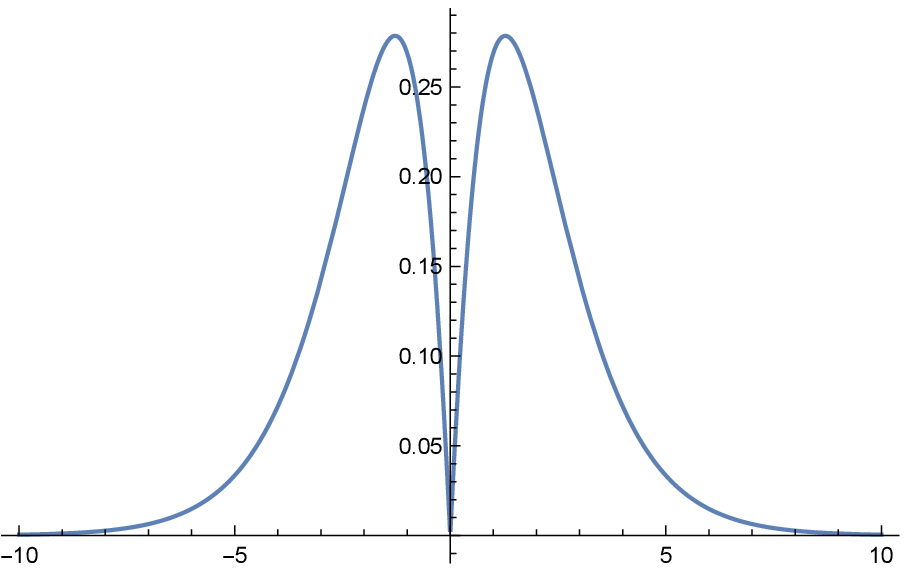}
\caption{The image of $u_{0}(x)$}
\end{subfigure}
\begin{subfigure}{0.49\textwidth}
\includegraphics[width=0.9\linewidth, height=5cm]{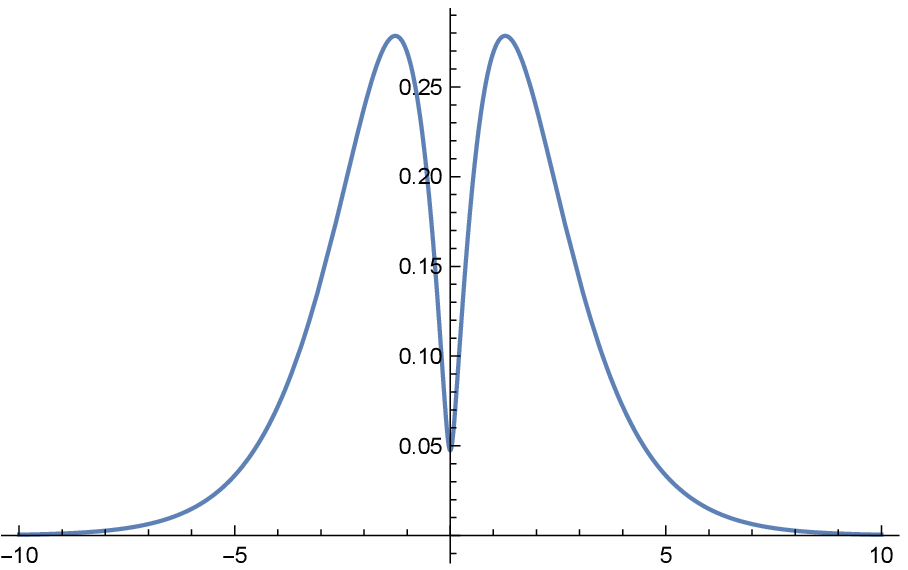}
\caption{The image of $u_{-0.1}(x)$}
\end{subfigure}

\label{fig:image2}
\end{figure}

When $a<0$, the spectral action coefficient of $t^a$ is given by
\begin{equation*}
\omega_{\mu}(a)
=
\int_0^\infty r_\mu(t)t^a dt
=
\frac{1}{\Gamma(-a)}\int_0^\infty u_\mu(\sqrt{x})x^{-a-1}dx.
\end{equation*} 
As before, we can express $\omega_\mu(a)$ in terms of the modified Bessel functions of the second kind:

\begin{lemma}\label{lemma_fermi_average_energy_1nd_eq}
For any fixed chemical potential $\mu<0$, the function $\omega_\mu(a)$ is given by
\begin{equation}\label{formula_1st_fermi_average_energy}
\omega_\mu(a)
=
\frac{2|2\mu|^{-a-\frac{1}{2}}}{\sqrt{\pi}}\sum_{n=1}^\infty (-1)^{n}\left(2a n^{a-\frac{1}{2}}K_{-a+\frac{1}{2}}\left(n|\mu|\right)- n^{a+\frac{1}{2}}|\mu| K_{-a-\frac{1}{2}}\left(n|\mu|\right)\right).
\end{equation}
Moreover, it can be extended to an entire function in $a$.
\end{lemma}

\begin{proof}
Taking any $\mu<0$, and using the same argument as in the proof of Proposition \ref{thm_fermi_coeff}, we can show that $\omega_\mu(a)$ can be extended to an entire function.
\end{proof}

A more explicit expression for $\omega_\mu(a)$ can again be found using the Poisson summation formula.

\begin{lemma}\label{lemma_fermi_average_energy_2nd_eq}
For any fixed chemical potential $\mu<0$,
we can express  $\omega_\mu(a)$ as
\begin{equation}\label{formula_2st_fermi_average_energy}
\omega_\mu(a)
=
\Gamma(a+1)\sum_{n=-\infty}^{\infty}\frac{(2n+1)^2\pi^2}{((2 n+1)^2\pi^2+\mu^2)^{a+1}}
-
\frac{|\mu|^{1-2a}}{4\sqrt{\pi}}\Gamma\left(a-\frac{1}{2}\right).
\end{equation}
\end{lemma}

\begin{proof}
Using   \eqref{Alt_0_order_Fourier_transform} and  applying Poisson's summation formula, we obtain, for any $\nu>0$ and $z>0$,
\begin{equation*}
\sum_{n=1}^\infty (-1)^n n^\nu K_{\nu}(z\, n)=\frac{\sqrt{\pi}}{2} \Gamma\left(\nu+\frac{1}{2}\right)(2z)^\nu\sum_{n=-\infty}^\infty\frac{1}{((2\pi n + \pi)^2+z^2)^{\nu+\frac{1}{2}}}-\frac{\Gamma(\nu)}{4}\left(\frac{2}{z}\right)^\nu.
\end{equation*} 
When $a>\frac{1}{2}$, 
if we apply this formula to \eqref{formula_1st_fermi_average_energy},
we can deduce the equation \eqref{formula_2st_fermi_average_energy}.
Now since $\omega_\mu(a)$ is an entire function, we conclude that  \eqref{formula_2st_fermi_average_energy} is valid in the whole complex plane.
\end{proof}
Moreover,
when $\mu\in(-\pi,0)$,
using a similar argument as in the proof of Lemma \ref{prop_gamma_mu_xi}
we can express the function $\omega_{\mu}$ in terms of $\xi$ function as well.

\begin{lemma}
When $\mu\in (-\pi,0)$,
we obtain the equation
\begin{equation*}
\omega_{\mu}(a)
=
\sum_{k=0}^\infty(-1)^k\frac{2(1-2^{-2(a+k)})\xi(2(a+k))}{(2a+2k-1)k!\pi^{a+k}}\mu^{2k}
-
\frac{\Gamma(a-1/2)}{4\sqrt{\pi}}|\mu|^{1-2a}.
\end{equation*}
\end{lemma}

If we write
\begin{equation*}
\omega(a):=\frac{2(1-2^{-2a})\xi(2a)}{(2a-1)\pi^{a}},
\end{equation*}
we can then express $\omega_\mu(a)$ as
\begin{equation}\label{eq_espresses_omega_mu}
\omega_\mu(a)
=
\sum_{k=0}^\infty(-1)^k\frac{\omega(a+k)}{k!}\mu^{2k}
-
\frac{\Gamma(a-1/2)}{4\sqrt{\pi}}|\mu|^{1-2a},\quad \mu\in(-\pi,0),
\end{equation}
Hence with the assumption of a heat expansion of the form \eqref{eq_ht},
we can express the average energy as follows:
\begin{prop}\label{prop_ave_eng_D_1_mu}
Assume we have the heat trace expansion $\tr(e^{-tD^2})\sim \sum_{l}\rho_l(t)$,
with respect to the scale $r_0 < r_1 < \cdots< r_l <\cdots$ we have
$\rho_l(t)=o_0(t^{r_{l}})$ and $\rho_l(t)=o_\infty(t^{r_l+1})$.
If we write $\psi_l(\beta, \mu):=\frac{1}{\beta}\sum\limits_{z\in X_l}\beta^{-2z}\omega_{\beta\mu}(-z)$
then we have the following asymptotic expansion for the average energy $\langle K_{f,\mu}\rangle_\beta$:
$$
\langle K_{f,\mu}\rangle_\beta = \frac{1}{\beta}\tr(u_{\beta\mu}(\beta D))
\sim
\sum_{l}\psi_l(\beta,\mu),\quad |\beta\mu|\to 0^+, |\mu|\to \infty,
$$
in the sense that 
$$
\frac{1}{\beta}\tr(u_{\beta\mu}(\beta D))
-
\sum_{l=0}^L\psi_l(\beta, \mu)
=
o_0(\beta^{2r_L-1}\omega_{\beta\mu}(r_L)),
\quad 
|\beta\mu|\to 0^+,\quad |\mu|\to \infty,
$$
and each term $\psi_l(\beta, \mu)=o_0(\beta^{2r_l-1}\omega_{\beta\mu}(r_l))$.
\end{prop}

\begin{proof}
By definition,
\begin{equation*}
\begin{aligned}
\langle K_{f,\mu}\rangle_\beta
&=
\frac{1}{\beta}\tr(u_{\beta\mu}(\beta D))
=
\frac{1}{\beta}\int_0^\infty r_{\beta\mu}(t)\tr(e^{-t\beta^2 D^2})dt,
\end{aligned}
\end{equation*}
and 
$$
\tr(e^{-t\beta^2D^2})
=
\sum_{l=0}^L\rho_{l}(t\beta^2) + R_L(t\beta^2),
$$
where the remainder term $R_L(t\beta^2)=o_0((t\beta^2)^{r_L})$,
and $R_L(t\beta^2)=o_\infty((t\beta^2)^{r_{L+1}})$.
we derive that
\begin{equation*}
\frac{1}{\beta}\tr(u_{\beta\mu}(\beta D))
=
\frac{1}{\beta}\sum_{l=0}^L\int_0^\infty r_{\beta\mu}(t)\rho_l(t\beta^2)dt 
+
\frac{1}{\beta}\int_0^\infty r_{\beta\mu}(t)R_L(t\beta^2))dt.
\end{equation*}
Recall that each term $\rho_l(t\beta^2)=\sum\limits_{z\in X_l}a_z(t\beta^2)^{-z}$ so that
$\int_0^\infty r_{\beta\mu}(t)\rho_l(t\beta^2)dt
=
\sum\limits_{z\in X_l}\beta^{-2z}\omega_{\beta\mu}(-z).
$
Now we need to show that the remainder 
$$
\frac{1}{\beta}\int_0^\infty r_{\beta\mu}(t)R_L(t\beta^2))dt=o_0(\beta^{2r_L-1}\omega_{\beta\mu}(r_L)), \textrm{ as } |\beta\mu|\to 0, \,\, |\mu|\to \infty.
$$ 
Since the remainder term $R_L(t\beta^2)=o_0((t\beta^2)^{r_L})$ and $R_L(t\beta^2)=O_\infty((t\beta^2)^{r_{L+1}})$,
using a similar argument as in the proof of \cite[Theorem 3.20]{Eckstein:2019dcb},
for any $\epsilon > 0$, there is an $M > 0$ such that 
\begin{equation}\label{eq_remainder_omega}
\left|\frac{1}{\beta}\int_0^\infty r_{\beta\mu}(t)R_L(t\beta^2))dt\right|
\leq 
\epsilon \beta^{2r_L-1} \omega_{\beta\mu}(r_{L})
+
M \beta^{2r_{L+1}-1}\omega_{\beta\mu}(r_{L+1}).
\end{equation}
We only need to show that
$$
\beta^{2(r_{L+1}-r_L)}\frac{\omega_{\beta\mu}(r_{L+1})}{\omega_{\beta\mu}(r_L)}\to 0, \textrm{ as}\quad |\mu|\to 0, \,\,|\beta\mu|\to 0.
$$
In fact by the formula \eqref{eq_espresses_omega_mu},
when $|\beta\mu|\to 0$ and $|\mu|\to \infty$,
$$
\beta^{2(r_{L+1}-r_L)}\frac{\omega_{\beta\mu}(r_{L+1})}{\omega_{\beta\mu}(r_L)}
\sim
\beta^{2(r_{L+1}-r_L)}\frac{\omega(r_{L+1})-\frac{\Gamma(r_{L+1}-1/2)}{4\sqrt \pi}|\beta\mu|^{1-2r_{L+1}}}{\omega(r_{L})-\frac{\Gamma(r_{L}-1/2)}{4\sqrt \pi}|\beta\mu|^{1-2r_{L}}},
$$
if $1-2r_{L+1} < 1-2r_L < 0,$
$$
\beta^{2(r_{L+1}-r_L)}\frac{\omega_{\beta\mu}(r_{L+1})}{\omega_{\beta\mu}(r_L)}
\sim
\frac{\Gamma(r_{L+1}-1/2)}{\Gamma(r_L-1/2)}|\mu|^{2(r_L-r_{L+1})}\sim 0, 
$$
if $0 < 1-2r_{L+1} < 1-2r_{L}$,
$$
\beta^{2(r_{L+1}-r_L)}\frac{\omega_{\beta\mu}(r_{L+1})}{\omega_{\beta\mu}(r_L)}
\sim
\beta^{2(r_{L+1}-r_L)}\frac{\omega(r_{L+1})}{\omega(r_L)}
\sim
0, 
$$
if $1-2r_{L+1} < 0 < 1-2r_L$,
$$
\beta^{2(r_{L+1}-r_L)}\frac{\omega_{\beta\mu}(r_{L+1})}{\omega_{\beta\mu}(r_L)}
\sim
\frac{\omega(r_{L+1})}{\omega(r_L)}\beta^{2(r_{L+1}-r_L)}-\frac{\Gamma(r_{L+1}-1/2)}{4\sqrt \pi \omega(r_L)}\beta^{1-2r_L}|\mu|^{1-2r_{L+1}}
\sim 0,
$$
hence
$$
\frac{1}{\beta}\int_0^\infty r_{\beta\mu}(t)R_L(t\beta^2)dt
=
o_0(\beta^{2r_L-1} \omega_{\beta\mu}(r_{L})).
$$
Since  each term $\rho_l(t\beta^2)=o_0((t\beta^2)^{r_{l}})$ and $\rho_l(t\beta^2)=o_\infty((t\beta^2)^{r_{l+1}})$,
using a similar argument we obtain that
\[
\psi_l(\beta, \mu)
=
\frac{1}{\beta}\int_{0}^\infty \rho_{l}(t\beta^2)r_{\beta\mu}(t)dt
=
o_0(\beta^{2r_l-1}\omega_{\beta\mu}(r_l)).
\qedhere
\]
\end{proof}

\subsection{The entropy of $\rho_f'$}
We now consider the function $h'_\mu$ defined by 
\begin{equation*}
h_\mu'(x) = \frac{|x|-\mu}{e^{|x|-\mu}+1}+\log\left(1+e^{-(|x|-\mu)}\right).
\end{equation*}
The entropy of the density operator $\rho'_f$ associated to $K_{f,\mu}'$ can then be expressed as
\begin{equation*}
\mathcal{S}(\rho_f')
=
-\tr\left(\rho_{f}'\log{\rho_{f}'}\right)
=
\tr(h_{\beta\mu}'(\beta D)).
\end{equation*}
Notice that 
\begin{equation*}
\tr(h_{\beta\mu}'(\beta D))\label{eq_spectral_action_h_2}
=
\int_0^\infty\Tilde{g}_{\beta\mu}(t)\tr\left(e^{-t\beta^2(D^2-2\mu|D|)}\right)dt,
\end{equation*}
where $\Tilde{g}_{\beta\mu}(t)=\Tilde{g}(t)e^{-t\beta^2\mu^2}.$

We can obtain the asymptotic expansion of $\tr(h_{\beta\mu}'(\beta D))$:

\begin{prop}
With the formula \eqref{eq_ht} and \eqref{eq_rho_l_k} we write
\begin{equation}
\begin{aligned}
\psi_{l,k}(\beta)
&=
\sum_{\substack{z\in X_l\\z+k/2\notin \mathbb{Z}^-}}\frac{1}{2}\Gamma(z+k/2)\Res(\zeta_D(s), z)\gamma_{\beta\mu}(-z+k/2)\beta^{-2z+k}\\ 
&+
\sum_{\substack{z\in X_l\\z+k/2\in \mathbb{Z}^-}}\frac{(-1)^{z+k/2}}{(-(z+k/2))!}\zeta_D(2z)\gamma_{\beta\mu}(-z+k/2)\beta^{-2z+k} . 
\end{aligned}
\end{equation}
We then have the following asymptotic expansion:
\begin{equation}\label{eq_asymp_h_2}
\tr(h_{\beta\mu}'(\beta D))
\sim
\sum_{k,l\geq 0}\frac{(2\mu)^k}{k!}\psi_{l,k}(\beta), \quad \beta\to 0^+,
\end{equation}
and each term $\psi_{l,k}(\beta)=o_0(\beta^{2r_l+k}).$
\end{prop}

\begin{remark}
Similar to the heat trace expansion \eqref{eq_ht_mu},
the asymptotic formula \eqref{eq_asymp_h_2} should be read in the sense that
$$
\tr(h_{\beta\mu}' (\beta D))
-
\sum_{\substack{0\leq k\leq K\\r_{l+1} < r_0+\frac{K-k}{2}}}\frac{(2\mu)^k}{k!}\psi_{l,k}(\beta)
=
o_0(\beta^{2r_0+K}).
$$
\end{remark} 

\subsection{The average energy of $K_{f,\mu}'$}
The average energy of $K_{f,\mu}'$ is given by
\begin{equation*}
\langle K_{f,\mu}' \rangle_\beta = \langle d\Gamma H_{f,\mu}'\rangle_\beta  
=
\tr\left(\frac{H_{f,\mu}'}{1+e^{\beta H_{f,\mu}'}}\right)
.
\end{equation*}
If we define
\begin{equation*}
u_\mu'(x)=\frac{|x|-\mu}{1+e^{|x|-\mu}}.
\end{equation*}
then it follows that
\begin{equation*}
\langle d\Gamma H_{f,\mu}'\rangle_\beta
=
\frac{1}{\beta}\tr(u_{\beta\mu}'(\beta D))
=
\frac{1}{\beta}\int_0^\infty r_{\beta\mu}(t)\tr(e^{-t\beta^2(D^2-2\mu |D|)})dt.
\end{equation*}

\begin{prop}
With the formula \eqref{eq_ht} and \eqref{eq_rho_l_k},
we denote by 
$$\psi_{l,k}(\beta,\mu)=\frac{1}{\beta}\int_0^\infty (t\beta^2)^k\rho_{l,k}(t\beta^2) r_{\beta\mu}(t)dt,$$
which can be explicitly expressed as
\begin{equation}\label{eq_psi_l_k}
\begin{aligned}
\psi_{l,k}(\beta, \mu)
&=
\sum_{\substack{z\in X_l\\z+k/2\notin \mathbb{Z}^-}}\Gamma(z+k/2)\Res(\zeta_{D^2}(s), z)\omega_{\beta\mu}(-z+k/2)\beta^{-2z+k-1}\\ 
&+
\sum_{\substack{z\in X_l\\z+k/2\in \mathbb{Z}^-}}\frac{(-1)^{z+k/2}}{(-(z+k/2))!}\zeta_{D^2}(z)\omega_{\beta\mu}(-z+k/2)\beta^{-2z+k-1},
\end{aligned}
\end{equation}
We have the following asymptotic expansion:
\begin{equation*}
\langle d\Gamma H_{f,\mu}'\rangle_\beta
=
\frac{1}{\beta}\tr(u_{\beta\mu}'(\beta D))
\sim
\sum_{k,l\geq 0} \frac{(2\mu)^k}{k!}\psi_{l,k}(\beta, \mu),\quad |\beta\mu|\to 0,\,\, |\mu|\to \infty,
\end{equation*}
and each term 
$$
\frac{(2\mu)^k}{k!}\psi_{l,k}(\beta, \mu)
=
o_0(\beta^{2r_l+k-1}\omega_{\beta\mu}(r_l+k/2)|\mu|^k).
$$
\end{prop}
\begin{remark}
The asymptotic expansion is such that
$$
\frac{1}{\beta}\tr(u_{\beta\mu}'(\beta D))
-
\sum_{\substack{0\leq k\leq K\\r_{l+1}\leq r_0+\frac{K-k}{2}}} \frac{(2\mu)^k}{k!}\psi_{l,k}(\beta, \mu)
=
o_0\left(\beta^{2r_0+K}\omega_{\beta\mu}(r_0+\frac{K+1}{2})|\mu|^{K+1}\right).
$$
\end{remark}

\begin{proof}
By definition
$$
\langle d\Gamma H_{f,\mu}'\rangle_\beta
=
\frac{1}{\beta}\tr(u_{\beta\mu}'(\beta D))
=
\frac{1}{\beta}\int_0^\infty \tr(e^{-t\beta^2(|D|-\mu)^2})r(t)dt.
$$
We denote the remainder by $R_K(t)$ :
$$
R_K(t)
=
\tr(e^{-t\beta^2(|D|-\mu)^2})
-
e^{-t\beta^2\mu^2}\sum_{\substack{0\leq k\leq K\\r_{l+1}\leq r_0+\frac{K-k}{2}}}\frac{(2t\beta^2\mu)^k}{k!}\rho_{l,k}(t\beta^2),
$$
we observe that $R_K(t) = o_0((t\beta^2)^{r_0+K/2})$, 
and $R_K(t) =o_\infty((t\beta^2)^{r_1+K/2})$,
hence using a similar argument as in the proof of Proposition \ref{prop_ave_eng_D_1_mu} we can show that
$$
\frac{1}{\beta}\int_0^\infty R_K(\beta, \mu)r(t\beta^2)dt
=
o_0\left(\beta^{2r_0+K}\omega_{\beta\mu}(r_0+\frac{K+1}{2})|\mu|^{K+1}\right).
$$
From the explicit expression of $\rho_{l,k}(t)$  given in \eqref{eq_rho_l_k} it follows that
\begin{multline*}
\rho_{l,k}(t)=\sum_{\substack{z\in X_l\\z+k/2\notin \mathbb{Z}^-}}\Gamma(z+k/2)\Res(\zeta_{D^2}(s), z)t^{-z-k/2} \\
+
\sum_{\substack{z\in X_l\\z+k/2\in \mathbb{Z}^-}}\frac{(-1)^{z+k/2}}{(-(z+k/2))!}\zeta_{D^2}(z)t^{-z-k/2}
,
\end{multline*}
Since this is a finite sum we may compute the integral term-by-term to obtain the equation \eqref{eq_psi_l_k}.
Since
$\rho_{l,k}(t\beta^2) = o_0((t\beta^2)^{r_l-k/2})$ and $\rho_{l,k}(t\beta^2) = o_\infty((t\beta^2)^{r_{l+1}-k/2})$,
we obtain that 
\begin{align*}
\frac{(2\mu)^k}{k!}\psi_{l,k}(\beta, \mu)
&=
\frac{(2\mu)^k}{k!}\int_0^\infty(t\beta^2)^k\rho_{l,k}(t\beta^2)r_{\beta\mu}(t)dt\\
&=
o_0\left(\beta^{2r_l+k-1}\omega_{\beta\mu}(r_l+k/2)|\mu|^k\right).\qedhere
\end{align*}
\end{proof}

\section{Bosonic second quantization}\label{sect_bose}

In this section we again assume $D:\mathcal{H}\to \mathcal{H}$ is a self adjoint operator with compact resolvent.

Let $\beta>0$, $\mu<0$ be the inverse temperature and chemical potential, respectively.
As before we denote by $H_{b,\mu}:=\sqrt{D^2+\mu^2\id}$ and $H_{b,\mu}':=D^2-\mu\id$ and write
\begin{equation*}
K_{b,\mu} := d\Gamma H_{b,\mu}: \bose\to \bose,
\quad
K_{b,\mu}' :=d\Gamma H_{b,\mu}': \bose\to \bose
\end{equation*}
for the corresponding second-quantized operators on the bosonic Fock space $\bose$.
When $e^{-\beta H_{b,\mu}}$ and $e^{-\beta H_{b,\mu}'}$ are trace class operators, we denote by
\begin{equation*}
Z_{b}=\tr(e^{-\beta K_{b,\mu}}),
\quad
Z_{b}'=\tr(e^{-\beta K_{b,\mu}'})
\end{equation*}
the corresponding canonical partition functions, and by
$$
\begin{aligned}
\rho_{b}
&:=
\frac{e^{-\beta K_{b,\mu}}}{Z_{b}},
\quad 
\phi_{b}(\cdot)
:=
\tr(\rho_{b}\cdot),\\
\rho_{b}'
&:=
\frac{e^{-\beta K_{b,\mu}'}}{Z_{b}'},
\quad 
\phi_{b}'(\cdot)
:=
\tr(\rho_{b}'\,\cdot)
\end{aligned}
$$
the corresponding density operators and Gibbs states. 

\subsection{The entropy of $\rho_b$}
We define a  function $k(x)$ by
\begin{equation}\label{eq_k}
k(x)
:=
-\frac{x}{1-e^{x}}-\log\left(1-e^{-x}\right).
\end{equation}
If we denote by $k_{\mu}(x):=k(\sqrt{x^2+\mu^2})$,
then we get
\begin{equation*}
\mathcal{S}(\phi_b)
=
-\tr(\rho_b\log\rho_b)
=
\tr\left(k_{\beta\mu}(\beta D)\right).
\end{equation*}

\begin{lemma}
The function $k(x)$ is an even positive function of the variable $x\in \mathbb{R}\backslash \{0\}$, and its derivative is
\begin{equation*}
\frac {dk }{dx}(x)=-\frac{x}{4\sinh^2(\frac{x}{2})}.
\end{equation*}
\end{lemma}

\begin{figure}[h]

\begin{subfigure}{0.49\textwidth}
\includegraphics[width=0.9\linewidth, height=5cm]{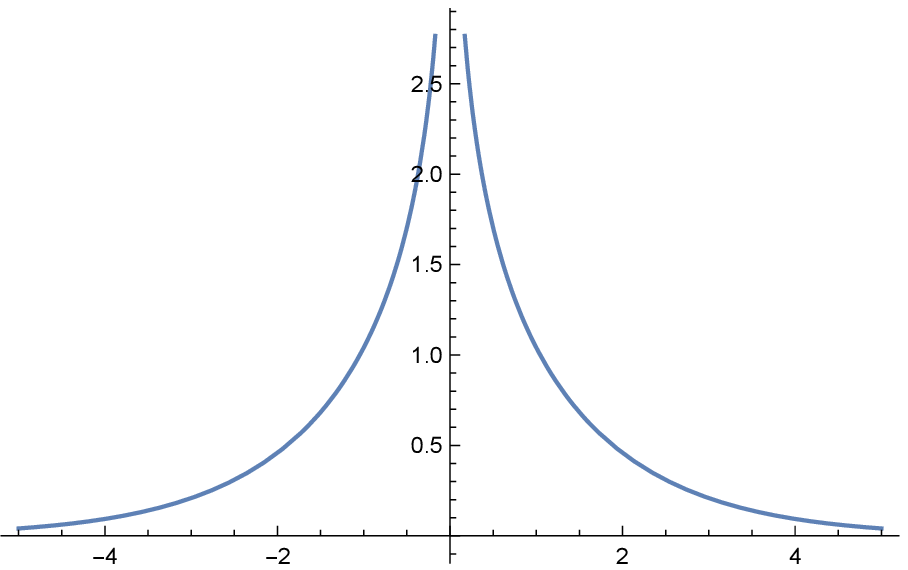}
\caption{The image of $k_{0}(x)$}
\end{subfigure}
\begin{subfigure}{0.49\textwidth}
\includegraphics[width=0.9\linewidth, height=5cm]{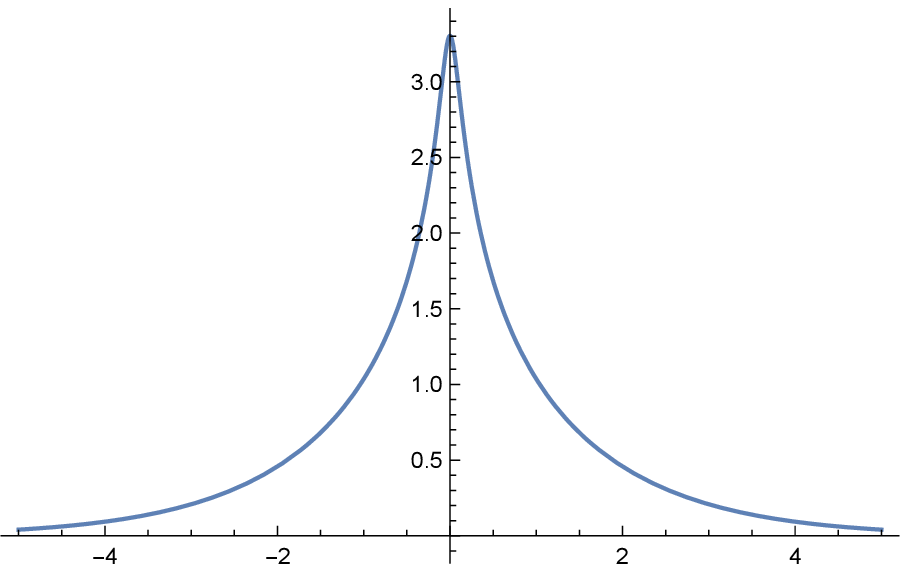}
\caption{The image of $k_{-0.1}(x)$}
\end{subfigure}

\label{fig:image2_}
\end{figure}

\begin{lemma}
For $x >0$, 
\begin{equation*}
\sum_{\mathbb{Z}}\frac{(2\pi n)^2-x}{((2\pi n)^2+x)^2}=-\frac{1}{4\sinh^2(\frac{\sqrt{x}}{2})}.
\end{equation*}
\end{lemma}
\begin{proof} 
We use the Eisenstein series \cite{CCS18} 
\begin{equation*}
\sum_{\mathbb{Z}}\frac{1}{(\pi n + x)^2}=\frac{1}{\sin^2 x},
\end{equation*}
in conjunction with 
\begin{equation*}
\sinh x =-i\sin(ix).
\end{equation*}
Thus
\begin{equation*}
\begin{aligned}
\frac{1}{4\sinh^2(\frac{\sqrt{x}}{2})}
=
-\frac{1}{4\sin^2(i\frac{\sqrt{x}}{2})}
=
-\sum_{\mathbb{Z}}\frac{1}{4(\pi n + i\frac{\sqrt{x}}{2})^2}
=
-\sum_{\mathbb{Z}}\frac{(2\pi n)^2-x}{((2\pi n)^2+x)^2}.
\end{aligned}
\end{equation*}
\end{proof}

Now since one has the equation
\begin{equation*}
\int_0^\infty \left(2(2\pi n)^2t-1\right)e^{-(2\pi n)^2t-tx}dt=\frac{(2\pi n)^2-x}{\left((2\pi n)^2+x\right)^2},
\end{equation*}
by the Fubini theorem we have the formula
\begin{equation*}
-\frac{1}{4\sinh^2{\frac{\sqrt{x}}{2}}}=\int_0^\infty f(t)e^{-tx}dt,\quad f(t):=\sum_{n\in \mathbb{Z}}\left(2(2\pi n)^2t-1\right)e^{-(2\pi n)^2t},
\end{equation*}
when $x>0$.

Next, let us determine the behaviour of the function $f(t)$ as $t\to0^+$ and $t\to+\infty$.
\begin{lemma}\label{lma}
Let
\begin{equation*}
f(t)=\sum_{n\in \mathbb{Z}}\left(2(2\pi n)^2t-1\right)e^{-(2\pi n)^2t}. 
\end{equation*}
The function $f(t)$ is rapidly decreasing as $t\to 0^+$ and approaches $-1$ as $t\to+\infty$.
\end{lemma}

\begin{proof}
Consider the theta function 
$\theta(t)=\sum\limits_{n\in \mathbb{Z}}e^{-\pi n^2t}.$
Let 
$g(t)=-2t\, \theta^\prime(t)-\theta(t).$
We have $f(t)=g(4\pi t).$
Thus it suffices to show that $g(t)$ is rapidly decreasing as $t\to 0^+$ and approaches $-1$ as $t\to+\infty$. 
Since $\theta'(t)$ is rapidly decreasing when $t\to +\infty$,
and $\theta(t)\to 1$ as $t\to+\infty$,
we proved that $g(t)$ approaches $-1$ as $t\to+\infty$.

Now,
using the Jacobi inversion formula,
$\theta(t)=\frac{1}{\sqrt{t}}\theta\left(\frac{1}{t}\right),$
we have
\begin{align*}
g(t)
&=
-2t\left(-\frac{1}{2}t^{-3/2}\theta\left(\frac{1}{t}\right) - \frac{1}{\sqrt t}\theta^\prime\left(\frac{1}{t}\right)t^{-2}\right)-\theta(t)\nonumber\\
&=
t^{-1/2}\theta\left(\frac{1}{t}\right)+2t^{-3/2}\theta^\prime\left(\frac{1}{t}\right) - \theta (t)\nonumber\\
&=
2t^{-3/2}\theta^\prime\left(\frac{1}{t}\right).
\end{align*}
Since as $t\to 0^+$, $\theta^\prime\left(\frac{1}{t}\right)$ is rapidly decreasing,
$g(t)$ is rapidly decreasing,
and also the function $f(t)$ is rapidly decreasing as $t\to 0^+.$
\end{proof}

Then we have the following lemma:
\begin{lemma}
When $x>0$, one has
\begin{equation}\label{k0_formula}
k(x)=\int_0^\infty e^{-tx^2}\Tilde{f}(t)dt, 
\end{equation}
here
$$\Tilde{f}(t)
=
\frac{f(t)}{2t}
=
\frac{1}{2t}\sum_{n\in \mathbb{Z}}\left(2(2\pi n)^2t-1\right)e^{-(2\pi n)^2t}.
$$
\end{lemma}

\begin{proof}
According to Lemma \ref{lma}, $\Tilde{f}(t)$ is rapidly decreasing as $t\to 0^+$. 
Thus when $x>0$, the integral on the right hand side is well-defined. We denote the integral on the right hand side of \eqref{k0_formula} by $\Tilde{k}(x)$. We have
\begin{equation*}
\partial_x \Tilde{k}(x)=-2x \int_0^\infty e^{-tx^2}t\Tilde{f}(t)dt=\frac{x}{4\sinh^2{\frac{x}{2}}}=-\partial_x k(x),
\end{equation*}
and since both $k(x)$ and $\Tilde{k}(x)$ approach to $0$ when $x\rightarrow \infty$, thus $k_0(x)=\Tilde{k}(x)$.
\end{proof}

Thus immediately we have 
\begin{prop}
When the chemical potential $\mu<0$, for all $x\in \mathbb{R}$, 
\begin{equation*}
k_{\mu}(x)=\int_0^\infty e^{-tx^2}\Tilde{f}_\mu(t)dt,
\end{equation*}
where 
$$\Tilde{f}_\mu(t)
=
e^{-\mu^2 t}\Tilde{f}(t)
=
\frac{e^{-\mu^2 t}}{2t}\sum_{n\in \mathbb{Z}}\left(2(2\pi n)^2t-1\right)e^{-(2\pi n)^2t}.
$$
\end{prop}

For the bosonic Fock space, we can get results for the moments of the function $k_\mu$ which are analogous those in Lemma \ref{fermimomentprop} for the fermionic case.
\begin{lemma}\label{bosemomentprop}
When $\nu>-1$, one has the integral formula
\begin{equation*}
\int_0^\infty k_{\mu}(x)x^\nu dx
=
|\mu|^{\frac{\nu}{2}+2} 2^{\frac{\nu}{2}}\frac{1}{\sqrt{\pi}}\Gamma\left(\frac{\nu+1}{2}\right)\sum_{n=1}^{\infty} n^{-\frac{\nu}{2}}K_{\frac{\nu}{2}+2}\left(n|\mu|\right).
\end{equation*}
\end{lemma}

Let us denote by $\chi_\mu(a)$  the $a$'th order spectral action coefficient of $k_{\mu}(\sqrt{x})$,
that is,
\begin{equation*}
\chi_\mu(a)
=
\int_0^\infty t^a\Tilde{f}_\mu(t)dt.
\end{equation*}
Analogous to Proposition \ref{thm_fermi_coeff} and \ref{fermi_spectral_action_coeff_2nd_prop}, we have the following lemma:
\begin{lemma}\label{sa_eqs_bose_entropy}
For a fixed chemical potential $\mu<0$, we can express the $a-$th order spectral action coefficient of $k_{\mu}(\sqrt{x})$ as:

\begin{equation}\label{eq_bose_coeff_entropy}
\chi_\mu(a)
=
\frac{1}{\sqrt{\pi}}2^{-a+\frac{1}{2}}|\mu|^{-a+\frac{3}{2}}\sum_{n=1}^\infty n^{a+\frac{1}{2}}K_{-a+\frac{3}{2}}\left(n|\mu|\right),
\end{equation}
and
\begin{equation}\label{2nd_eq_bose_coeff_entropy}
\chi_{\mu}(a)
=
-\frac{\Gamma(a)}{2}\sum_{n=-\infty}^\infty\frac{(2a-1)(2n)^2\pi^2-\mu^2}{((2n)^2\pi^2+\mu^2)^{a+1}}.
\end{equation}
Moreover,  the expressions \eqref{eq_bose_coeff_entropy} and \eqref{2nd_eq_bose_coeff_entropy} both are entire functions with respect to  $a\in \mathbb{C}$.
\end{lemma}

We can also express the function $\chi_\mu(a)$ in terms of the Riemann $\xi$ function when $|\mu|$ is small enough.
\begin{lemma}
For a fixed $\mu\in(-2\pi,0)$,
we have the formula
\begin{equation*}
\chi_\mu(a)=\frac{\Gamma(a)}{2}|\mu|^{-2a}+\sum_{n=0}^\infty(-1)^{n+1}\frac{\chi(n+a)}{n!}\mu^{2n},
\end{equation*}
where 
$\chi(a)=\frac{\xi(2a)}{(4\pi)^aa}.$
\end{lemma}

Hence with the assumption of the heat trace expansion \eqref{eq_ht},
using a similar argument as in the proof of Proposition \ref{prop_ave_eng_D_1_mu},
we finally obtain the asymptotic expansion of $\tr(k_{\beta\mu}(\beta D))$:

\begin{prop}
With the heat trace  expansion \eqref{eq_ht},
let 
\begin{equation*}
\psi_l(\beta, \mu):=\int_0^\infty \rho_{l}(t\beta^2)\Tilde{f}_{\beta\mu}(t)dt
=
\sum_{z\in X_l}a_z\beta^{-2z}\chi_{\beta\mu}(-z),
\end{equation*}
we obtain the following asymptotic expansion for the von Neumann entropy $\mathcal{S}(\rho_b)$:
\begin{equation*}
\mathcal{S}(\rho_b)
=
\tr(k_{\beta\mu}(\beta D))
\sim
\sum_{l}\psi_l(\beta, \mu),\quad |\beta\mu|\to 0,\,\, |\mu|\to \infty.
\end{equation*}
More precisely,
\begin{equation*}
\tr(k_{\beta\mu}(\beta D))
-
\sum_{l=0}^L\psi_l(\beta, \mu)
=
o_0(\beta^{2r_L}\chi_{\beta\mu}(r_L)),
\end{equation*}
and each term $\psi_l(\beta, \mu)=o_0(\beta^{2r_l}\chi_{\beta\mu}(r_l))$.
\end{prop}

\subsection{The average energy of $K_{b,\mu}$}
In the bosonic case, the average energy with respect to the operator $K_{b,\mu} = d\Gamma H_{b,\mu}$ is equal to 
\begin{equation*}
\langle d\Gamma H_{b,\mu} \rangle_\beta
=
-\frac{\partial}{\partial \beta}\left(\log Z_{b,\beta,\mu}\right)
=
-\tr\left(\frac{H_{b,\mu}}{1-e^{\beta H_{b,\mu}}}\right).
\end{equation*}
If we write $p_{\mu}(x):=-\frac{\sqrt{x^2+\mu^2}}{1-e^{\sqrt{x^2+\mu^2}}}$, then the average energy is equal to
\begin{equation*}
\langle d\Gamma H_{b,\mu} \rangle_\beta
=
\frac{1}{\beta}\tr(p_{\beta\mu}(\beta D)).
\end{equation*}
As with the discussion in section \ref{section_fermi_av_eng},
with  the chemical potential $\mu<0$,
the function $p_{\mu}(x)$ is given by the following Laplace transform:
\begin{equation}\label{integral}
p_{\mu}(\sqrt x)
=
\int_0^\infty s_\mu(t)e^{-tx}dt, \qquad \mu<0, \quad  x\geq 0,
\end{equation}
where
\begin{equation*}
s_\mu(t)
=
\frac{1}{\sqrt{\pi}} t^{-5/2}\sum_{n=1}^\infty \left(\frac{n^2}{4}-\frac{t}{2}\right)e^{-\frac{n^2}{4t}}e^{-\mu^2 t}.
\end{equation*}
In contrast to the case of fermionic second quantization, here we cannot take $\mu=0$, 
as then the integral on the right-hand side of the formula \eqref{integral} does not converge.
This is consistent with the fact that $p_{0}(x)$ is singular at $x=0$.

\begin{figure}[h]

\begin{subfigure}{0.49\textwidth}
\includegraphics[width=0.9\linewidth, height=5cm]{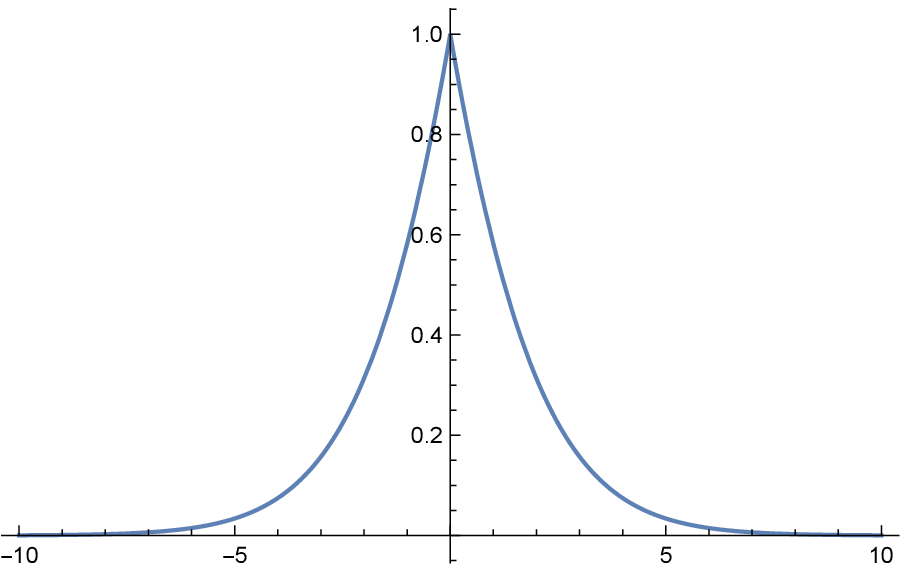}
\caption{The image of $p_{0}(x)$}
\end{subfigure}
\begin{subfigure}{0.49\textwidth}
\includegraphics[width=0.9\linewidth, height=5cm]{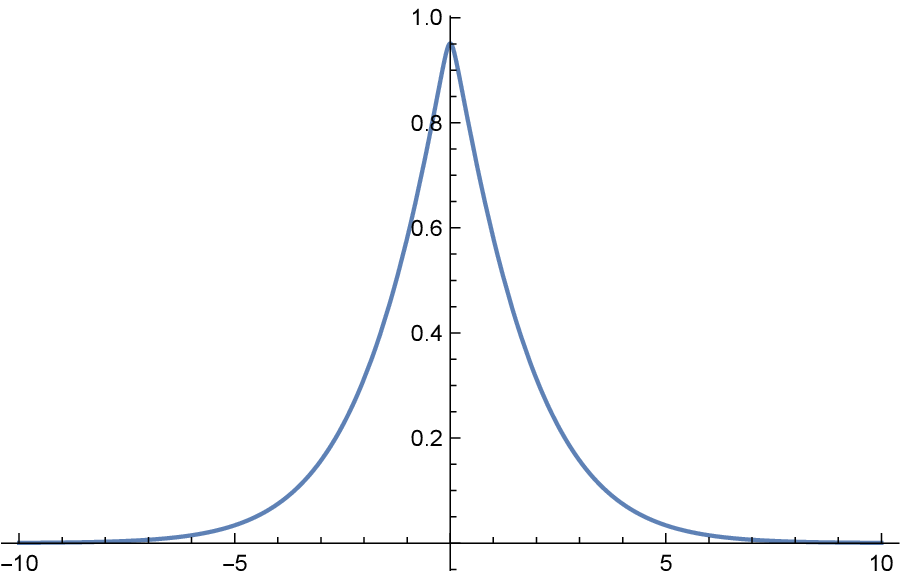}
\caption{The image of $p_{-0.1}(x)$}
\end{subfigure}

\label{fig:image2__}
\end{figure}

When the order $a<0$, we denote by $\kappa_\mu(a)$ the spectral action coefficient of the spectral action function $p_{\mu}(\sqrt{x})$,
namely,
\begin{equation*}
\kappa_{\mu}(a)
=
\int_0^\infty s_\mu(t)t^adt
=
\frac{1}{\Gamma(-a)}\int_0^\infty p_{\mu}(\sqrt{x})x^{-a-1}dx
=
\frac{2}{\Gamma(-a)}\int_0^\infty p_{\mu}(x)x^{-2a-1}dx.
\end{equation*}
Using the same argument as in section \ref{section_fermi_av_eng},
we have 
\begin{lemma}
For any fixed chemical potential $\mu<0$, 
\begin{equation}\label{formula_1st_bose_average_energy}
\kappa_\mu(a)
=
\frac{2^{-a+\frac{1}{2}}}{\sqrt{\pi}}|\mu|^{-a+\frac{1}{2}}\sum_{n=1}^\infty \left(n^{a+\frac{1}{2}}|\mu|\,  K_{-a-\frac{1}{2}}\left(n|\mu|\right)-2\,a\,n^{a-\frac{1}{2}}K_{-a+\frac{1}{2}}\left(n|\mu|\right)\right),
\end{equation}
and it can be extended to a holomorphic function on $ \mathbb{C}$.
Thus  this formula gives the spectral action coefficients of all orders.
Moreover,
\begin{equation}\label{formula_2st_bose_average_energy}
\kappa_\mu(a)
=
-\Gamma(a+1)\sum_{n=-\infty}^{\infty}\frac{(2n)^2\pi^2}{((2 n)^2\pi^2+\mu^2)^{a+1}}
+
\frac{|\mu|^{-2a+1}}{4\sqrt{\pi}}\Gamma\left(a-\frac{1}{2}\right), 
\end{equation}
which can also be extended to an entire function for any fixed chemical potential $\mu<0$.
\end{lemma}
Using the same trick as in the previous section,
we can express $\kappa_\mu(a)$ in terms of Riemann $\xi$ function as follows:
\begin{lemma}
Let 
\begin{equation*}
\kappa(a)
=
\frac{2\xi(2a)}{(4\pi)^a(2a-1)},
\end{equation*}
when $\mu\in(-2\pi,0)$,
we can express the function $\alpha_\mu(a)$ as
\begin{equation}\label{eq_kappa_mu}
\kappa_\mu(a)
=
\sum_{n=0}^\infty(-1)^{n+1}\frac{\kappa(a+n)}{n!}\mu^{2n}+\frac{|\mu|^{1-2a}}{4\sqrt{\pi}}\Gamma\left(a-\frac{1}{2}\right).
\end{equation}
\end{lemma}

\begin{prop}
In terms of the heat trace expansion \eqref{eq_ht} let 
\begin{equation*}
\psi_l(\beta, \mu)
:=
\frac{1}{\beta}\int_0^\infty \rho_{l}(t\beta^2)s_{\beta\mu}(t)dt
=
\frac{1}{\beta}\sum_{z\in X_l}a_z\beta^{-2z}\kappa_{\beta\mu}(-z).
\end{equation*}
We then obtain the asymptotic expansion for the average energy:
\begin{equation*}
\langle K_{b,\mu}\rangle_\beta
=
\frac{1}{\beta}\tr(p_{\beta\mu}(\beta D))
\sim
\sum_l\psi_l(\beta, \mu),\quad |\beta\mu|\to 0,\,\,\, |\mu|\to \infty,
\end{equation*}
in the sense that 
\begin{equation*}
\frac{1}{\beta}\tr(p_{\beta\mu}(\beta D))
-
\sum_{l=0}^L\psi_l(\beta, \mu)
=
o_0(\beta^{2r_L-1}\kappa_{\beta\mu}(r_L)),\end{equation*}
and each term $\psi_l(\beta, \mu)=o_0(\beta^{2r_l-1}\kappa_{\beta\mu}(r_l))$.
\end{prop}

\subsection{The entropy of $\rho_b'$}
Let $k_\mu'(x)=k(x-\mu)$, 
where the function $k(x)$ is given in \eqref{eq_k}.
Then the entropy of the density operator $\rho_{b}'$ is given by
\begin{equation*}
\mathcal{S}(\rho_b')
=
-\tr(\rho_b'\log\rho_b')
=
\tr(k_{\beta\mu}'(\beta D^2)).
\end{equation*}

Now we shall consider the spectral action
\begin{equation*}
\tr(k_{\beta\mu}(\beta D^2))
=
\int_0^\infty \Tilde{{f}}(t)\tr(e^{-t\beta^2(D^2-\mu)^2})dt.
\end{equation*}
\begin{prop}
With the asymptotic expansion \eqref{eq_ht_D4_mu},
let 
\begin{equation*}
\psi_{l,k}(\beta, \mu)=\int_0^\infty \Tilde{f}_{\beta\mu}(t)(t\beta^2)^k\Tilde{\rho}_{l,k}(t\beta^2)dt,
\end{equation*}
and, more explicitly,
\begin{equation*}
\begin{aligned}
\psi_{l,k}(\beta, \mu)
=
&\frac{1}{2}\sum_{\substack{z\in X_l\\(z+k)/2\notin \mathbb{Z}^-}}\Gamma\left(\frac{k+z}{2}\right)\Res(\zeta_{D^2}(s), z)\chi_{\beta\mu}\left(\frac{-z+k}{2}\right)\beta^{-z+k}\\
&+
\sum_{\substack{z\in X_l\\(z+k)/2\in \mathbb{Z}^-}}\frac{(-1)^{\frac{k+z}{2}}}{\left(-\frac{k+z}{2}\right)!}\zeta_{D^2}(z)\chi_{\beta\mu}\left(\frac{-z+k}{2}\right)\beta^{-z+k}.
\end{aligned}
\end{equation*}

Then there is the following asymptotic expansion:
\begin{equation*}
\mathcal{S}(\rho_b')
=
\tr(k_{\beta\mu}(\beta D^2))
\sim
\sum_{l,k}\frac{(2\mu)^k}{k!}\psi_{l,k}(\beta, \mu),\quad |\beta\mu|\to 0,\,\, |\mu|\to \infty,
\end{equation*}
More precisely,
\begin{equation*}
\tr(k_{\beta\mu}(\beta D^2))
-
\sum_{\substack{0\leq k\leq K\\r_{l+1} \leq r_0+K-k}}\frac{(2\mu)^k}{k!}\psi_{l,k}(\beta, \mu)
=
o_0\left(\beta^{r_0+K}\chi_{\beta\mu}\left(\frac{r_0+K}{2}\right)|\mu|^K\right),
\end{equation*}
and each term
\begin{equation*}
\frac{(2\mu)^k}{k!}\psi_{l,k}(\beta, \mu)
=
o_0\left(\beta^{r_{l}+k}\chi_{\beta\mu}\left(\frac{r_{l}+k}{2}\right)|\mu|^{k}\right).
\end{equation*}
\end{prop}

\subsection{The average energy of $K_{b,\mu}'$}
Let $p_{\mu}(x)=-\frac{x-\mu}{1-e^{x-\mu}}$ and recall that $H_{b,\mu}'=D^2-\mu\id$.
By definition,
the average energy of $K_{b,\mu} '= d\Gamma H_{b,\mu}'$
is given by
\begin{equation*}
\langle K_{b,\mu}'\rangle_\beta
=
-\frac{\partial}{\partial \beta}(\log Z_{b,\beta,\mu})
=
\frac{1}{\beta}\tr(p_{\beta\mu}(\beta D^2))
=
\frac{1}{\beta}\int_0^\infty s(t)\tr\left(e^{-t\beta^2(D^2-\mu)^2}\right)dt.
\end{equation*}
Using similar argument as before,
we obtain the following proposition:

\begin{prop}
With the asymptotic expansion \eqref{eq_ht_D4_mu},
let 
\begin{equation*}
\psi_{l,k}(\beta, \mu)
=
\frac{1}{\beta}\int_0^\infty s_{\beta\mu}(t)(t\beta^2)^k\Tilde{\rho}_{l,k}(t\beta^2)dt,
\end{equation*}
and, more explicitly,
\begin{equation*}
\begin{aligned}
\psi_{l,k}(\beta, \mu)
=
&\frac{1}{2}\sum_{\substack{z\in X_l\\(z+k)/2\notin \mathbb{Z}^-}}\Gamma\left(\frac{k+z}{2}\right)\Res(\zeta_{D^2}(s), z)\kappa_{\beta\mu}\left(\frac{-z+k-1}{2}\right)\beta^{-z+k-1}\\
&+
\sum_{\substack{z\in X_l\\(z+k)/2\in \mathbb{Z}^-}}\frac{(-1)^{\frac{k+z}{2}}}{\left(-\frac{k+z}{2}\right)!}\zeta_{D^2}(z)\kappa_{\beta\mu}\left(\frac{-z+k}{2}\right)\beta^{-z+k-1}.
\end{aligned}
\end{equation*}
We then obtain the asymptotic expansion:
\begin{equation*}
\langle K_{b,\mu}'\rangle_\beta
=
\frac{1}{\beta}\int_0^\infty s(t)\tr\left(e^{-t\beta^2(D^2-\mu)^2}\right)dt
\sim
\sum_{l,k}\frac{(2\mu)^k}{k!}\psi_{l,k}(\beta, \mu),\quad |\beta\mu|\to 0,\,\, |\mu|\to \infty.
\end{equation*}
More precisely, 
\begin{equation*}
\langle K_{b,\mu}'\rangle_\beta
-
\sum_{\substack{0\leq k\leq K\\r_{l+1} \leq r_0+K-k}}\frac{(2\mu)^k}{k!}\psi_{l,k}(\beta, \mu)
=
o_0\left(\beta^{r_0+K}\kappa_{\beta\mu}\left(\frac{r_0+K+1}{2}\right)|\mu|^K\right),
\end{equation*}
and each term
\begin{equation*}
\frac{(2\mu)^k}{k!}\psi_{l,k}(\beta, \mu)
=
o_0\left(\beta^{r_{l}+k-1}\kappa_{\beta\mu}\left(\frac{r_{l}+k}{2}\right)|\mu|^{k}\right).
\end{equation*}
\end{prop}

\appendix

\section{Modified Bessel functions of the second kind}
\label{sect_bessel}
The modified Bessel functions $\{I_\nu(z), K_\nu(z)\}$ are the solutions of the modified Bessel differential  equation
\begin{equation*}
z^2y^{\prime\prime} + zy^{\prime} - (z^2+\nu^2)y=0,
\end{equation*}
where 
\begin{equation*}
I_\nu(z)=\left(\frac{1}{2}z\right)^{\nu}\sum_{n=0}^\infty \frac{\left(\frac{1}{2}z\right)^{2n}}{\Gamma(n+\nu+1)n!},
\end{equation*}
and
\begin{equation}\label{K_fun}
K_\nu(z)=\frac{\pi}{2}\frac{I_{-\nu}(z)-I_\nu(z)}{\sin \nu\pi}, \quad -\pi<\textrm{arg}z<\pi.
\end{equation}
The right-hand side of \eqref{K_fun} should be determined by taking the limit when $\nu$ is an integer.
The function $I_\nu(z)$ is called the modified Bessel function of the first kind, 
and $K_\nu(z)$ the modified Bessel function of the second kind.

We shall introduce some basic properties of the modified Bessel function of the second kind referring to \cite{temme1996special, bateman1954tables, gradshteyn2007} for more details.

\begin{lemma}
When $\alpha\in \mathbb{R}$, one has the formula
$$K_\alpha(z)=K_{-\alpha}(z).$$
\end{lemma}

\begin{lemma}\label{0_1_bessel_lemma}
We have the formula
\begin{equation*}
\frac{d}{dz}K_0(z)=-K_1(z).
\end{equation*}
\end{lemma}

\begin{lemma}\label{lemma_bessel_asymp}
For $\alpha>0$, when $z\to 0^+,$
one has the asymptotics
\begin{equation*}
K_\alpha(z)\sim
\left\{
\begin{array}{ll}
-\log(\frac{z}{2})-\gamma \quad \alpha=0,\\
\frac{\Gamma(\alpha)}{2}\left(\frac{2}{z}\right)^\alpha \qquad \alpha > 0,
\end{array}
\right.
\end{equation*}
where $\gamma$ is Euler's constant.
When $z\nearrow \infty$,
one has
\begin{equation*}
K_\alpha(z)
\sim
\sqrt{\frac{\pi}{2z}}e^{-z}.
\end{equation*}
\end{lemma}

\begin{lemma}\label{intrep}
One has the integral representation formula of the function $K_\nu(z)$:
\begin{equation*}
K_\nu(z)=\frac{\sqrt{\pi}}{\Gamma\left(\nu+\frac{1}{2}\right)}\left(\frac{z}{2}\right)^\nu\int_1^\infty e^{-zx}(x^2-1)^{\nu-1/2}dx \qquad \textrm{for}\quad  \nu>-\frac{1}{2}.
\end{equation*}
\end{lemma}

\begin{lemma}{\cite[8.486]{gradshteyn2007}}
Let $K_\nu(z)$ be the modified Bessel function of the second kind.
Then one has 
\begin{equation}\label{bess1}
zK_{\nu-1}(z)-zK_{\nu+1}(z)=-2\nu K_{\nu}(z),
\end{equation}
\begin{equation}\label{bess2}
K_{\nu-1}(z)+K_{\nu+1}(z)=-2\frac{\partial}{\partial z}K_{\nu}(z),
\end{equation}

\begin{equation}\label{bess3}
z\frac{\partial}{\partial z}K_{\nu}(z)+\nu K_{\nu}(z)=-zK_{\nu-1}(z),
\end{equation}

\begin{equation}\label{bess4}
z\frac{\partial}{\partial z}K_{\nu}(z)-\nu K_\nu(z)=-zK_{\nu+1}(z).
\end{equation}
\end{lemma}

\begin{lemma}{\cite[8.432]{gradshteyn2007}}\label{8_432_gradshteyn2007}
When $\nu> -\frac{1}{2}$, $a>0$, and $x>0$, we have the integral formula 
\begin{equation*}
x^{\nu}K_{\nu}(ax)
=
\frac{\Gamma\left(\nu+\frac{1}{2}\right)(2a)^\nu}{\Gamma\left(\frac{1}{2}\right)}\int_0^\infty \frac{\cos xt}{(t^2+a^2)^{\nu+\frac{1}{2}}}dt.
\end{equation*}
\end{lemma}

Using Lemma \ref{8_432_gradshteyn2007}, 
we obtain the following Lemma:
\begin{lemma}\label{lemma_0_order_FT}
When $\nu> -\frac{1}{2}$, $a>0$, and $x\in \mathbb{R}\backslash \{0\}$, one has
\begin{equation}\label{0_order_Fourier_Transform}
|x|^\nu K_\nu\left(a|x|\right)
=
\frac{\pi\Gamma\left(\nu+\frac{1}{2}\right)(2a)^\nu}{\Gamma\left(\frac{1}{2}\right)}\widehat{\psi}_{\nu,a}(x),
\end{equation}
and
\begin{equation}\label{Alt_0_order_Fourier_transform}
e^{i\pi x}|x|^\nu K_\nu\left(a|x|\right)
=
\frac{\pi\Gamma\left(\nu+\frac{1}{2}\right)(2a)^\nu}{\Gamma\left(\frac{1}{2}\right)}\widehat{\phi}_{\nu,a}(x),
\end{equation}
where 
$$\psi_{\nu,a}(t)
=
\frac{1}{\left((2\pi t)^2+a^2\right)^{\nu+\frac{1}{2}}}, \quad
\phi_{\nu,a}(t)=\psi_{\nu,a}\left(t+\frac{1}{2}\right),$$
and $\widehat{\psi}_{\nu,a}$, $\widehat{\phi}_{\nu,a}$ denote the corresponding Fourier transforms of $\psi_{\nu,a}$ and $\phi_{\nu,a}$.
Namely,
$$\widehat{\psi}_{\nu,a}(x)=\int_{-\infty}^\infty \psi_{\nu,a}(t)e^{-2\pi ixt}dt, \quad \widehat{\phi}_{\nu,a}(x)=\int_{-\infty}^\infty \phi_{\nu,a}(t)e^{-2\pi ixt}dt.$$
\end{lemma}

\begin{proof}
According to Lemma \ref{8_432_gradshteyn2007}, one has 
\begin{equation*}
|x|^\nu K_\nu(a|x|)
=
\frac{\Gamma\left(\nu+\frac{1}{2}\right)(2a)^\nu}{2\Gamma\left(\frac{1}{2}\right)}\int_{-\infty}^\infty \frac{1}{\left(t^2+a^2\right)^{\nu+\frac{1}{2}}}e^{-ixt}dt,
\end{equation*}
and then changing the variable
$t\mapsto 2\pi t$,
we can obtain formulae \eqref{0_order_Fourier_Transform} and \eqref{Alt_0_order_Fourier_transform}.
\end{proof}
From Lemma \ref{lemma_0_order_FT},
one can easily deduce the following lemma:

\begin{lemma}\label{lemma_2_order_FT}
When $\nu\geq -\frac{1}{2}$, $a>0$, and $x\in \mathbb{R}\backslash \{0\}$, one has
\begin{equation}\label{1_order_Fourier_Transform}
|x|^{\nu+2} K_\nu\left(a|x|\right)
=
\frac{\Gamma\left(\nu+\frac{1}{2}\right)(2a)^\nu}{-4\pi\Gamma\left(\frac{1}{2}\right)}\widehat{\psi^{\prime\prime}}_{\nu,a}(x),
\end{equation}
and
\begin{equation}
e^{i\pi x}|x|{^{\nu+2}} K_\nu\left(a|x|\right)
=
\frac{\Gamma\left(\nu+\frac{1}{2}\right)(2a)^\nu}{-4\pi \Gamma\left(\frac{1}{2}\right)}\widehat{\phi^{\prime\prime}}_{\nu,a}(x).
\end{equation}
\end{lemma}

\section{Asymptotic expansion}\label{Appendix_asy_exp}
\label{sect_heat-exp}
We deduce, under suitable conditions,  a heat trace expansion for  $\tr(e^{-t(|D|-\mu)^2})$  and for 
$\tr(e^{-t(D^2-\mu)^2})$,  from the heat trace expansion of $\tr(e^{-tD^2})$.

Let $D$ be a self-adjoint unbounded operator  acting on a Hilbert space $\mathscr{H}$ such that $\ker D=0$. We assume that  for some $p>0,$  and all $\varepsilon >0,$
 $$\tr |D|^{-p -\varepsilon}<\infty, \,\, \textrm{but}\,\, \tr |D|^{-p + \varepsilon}=\infty.$$
We obtain the spectral zeta function
$$\zeta_{D^2}(s):=\tr|D|^{-2s}, \quad \re (s) > p/2,$$
and we denote by 
$$\mathcal{Z}_k(s):=\Gamma(s)\zeta_{D^2}(s-k/2),\quad k\in\mathbb{N}.$$


We assume that
\begin{enumerate}
    \item There exists a sequence $r_0<r_1<r_2<\cdots$ 
    strictly increasing to $+\infty$ and a discrete set $X\subset \mathbb{C}$ such that for each vertical strip $U_{l}:=\{z\in\mathbb{C}\Big | -r_{l+1}<\re(z)<-r_l\}$ the intersection $X_l:=U_{l}\cap X$ is a finite set and 
    \begin{equation}\label{eq_ht}
        \tr e^{-tD^2}\isEquivTo{t\downarrow 0}\sum_{l=0}^\infty \rho_l(t),\quad \textrm{with}\quad \rho_l(t)=\sum_{z\in X_l}a_{z}t^{-z}.
    \end{equation}
    \label{assp_1}
    
    
    \item $\zeta_{D^2}(s)$ is regular at $s=-n$ for all $n\in\mathbb{N}$.\label{assp_2}
\end{enumerate}

\begin{remark}
Recall from \cite[Theorem 3.2]{Eckstein:2019dcb} that by assumption (\ref{assp_1})
we may conclude that $\zeta_{D^2}(s)$ admits a meromorphic extension to the  whole complex plane with only simple poles in $\Sigma(D)$. Moreover
$$\Res(\mathcal{Z}_0(s), z)=a_{z}.$$
\end{remark}

By assumption \ref{assp_1}
$$\tr e^{-tD^2}
=
\sum_{l=0}^N\rho_l(t)+R_N(t),
$$
where $R_N(t)=O (t^{r_{N+1}})$ as $t\to 0^+$. 

Moreover, 
for each $k\in\mathbb{N}$
$$\tr (|D|^ke^{-tD^2})=\sum_{l=0}^N\rho_{l,k}(t)+R_{N,k}(t).$$
In more detail,
$\rho_{l,k}$ can be expressed as
\begin{eqnarray}\label{eq_rho_l_k}
\begin{split}
\rho_{l,k}(t) &=&\sum_{\substack{z\in X_l\\z+k/2\notin \mathbb{Z}^-}}\Gamma(z+k/2)\Res(\zeta_{D^2}(s), z)t^{-z-k/2} \\
&&+
\sum_{\substack{z\in X_l\\z+k/2\in \mathbb{Z}^-}}\frac{(-1)^{z+k/2}}{(-(z+k/2))!}\zeta_{D^2}(z)t^{-z-k/2}
,
\end{split}
\end{eqnarray}

so that 
$\rho_{l,k}(t)=o (t^{r_{l}-k/2})$,  as $t \to 0^+$, 
and $\rho_{l,k}(t)=o (t^{r_{l+1}-k/2})$, as  $t \to +\infty.$ \\

  Now we apply the functional calculus to show that 
\begin{equation}\label{eq_exp_expand}
\begin{aligned}
e^{-t(|D|-\mu)^2}
-
e^{-t\mu^2}\sum_{k=0}^K\frac{(2t\mu)^k}{k!}e^{-tD^2}|D|^k
\leq 
e^{-t\mu^2}\frac{|2t\mu|^{K+1}}{(K+1)!}e^{-tD^2}|D|^{K+1}
.
\end{aligned}
\end{equation}
 Indeed, we can consider the  Taylor expansion of $e^{-t(|\lambda|-\mu)^2}$ on   an eigenvalue $\lambda$ of $D$:
$$e^{-t(|\lambda|-\mu)^2} = e^{-t\mu^2}e^{-t\lambda^2}\left(\sum_{k=0}^K\frac{(2t\mu)^k}{k!}|\lambda|^k + \mathcal{R}_K(t\mu)\right),$$
where the remainder satisfies 
\begin{align*}
\mathcal{R}_K(t\mu)
&=
\frac{(-2|\lambda|)^{K+1}}{K!}\int_{t\mu}^{0}e^{2s|\lambda|}(s-t\mu)^Kds\\
&\leq %
\frac{(2|\lambda|)^{K+1}}{K!}
\int_{t\mu}^0(s-t\mu)^{K}ds
\\
&=
\frac{|2t\mu\lambda|^{K+1}}{(K+1)!}.\\
\end{align*}
Now the operator inequality   \eqref{eq_exp_expand} follows. Taking the trace on both sides of that inequality yields 
\begin{equation*}
\tr\left(e^{-t(|D|-\mu)^2}\right)
-
e^{-t\mu^2}\sum_{k=0}^K\frac{(2t\mu)^k}{k!}\tr\left(e^{-tD^2}|D|^k\right) 
\leq
e^{-t\mu^2}\frac{|2t\mu|^{K+1}}{(K+1)!}\tr\left(e^{-tD^2}|D|^{K+1}\right).
\end{equation*}
We now realize that for $k\leq K$
$$
\tr\left(e^{-tD^2}|D|^k\right)
=
\sum_{r_{l+1} \leq r_0+\frac{K-k}{2}}\rho_{l,k}(t)+o(t^{r_0+\frac{K}{2}-k}), \quad t\to 0^+,
$$
and 
$$
\tr\left(e^{-tD^2}|D|^{K+1}\right)
=
o (t^{r_0-\frac{K+1}{2}}), \quad t\to 0^+,
$$
hence we obtain the asymptotic expansion
\begin{equation}\label{eq_ht_mu}
\tr(e^{-t(|D|-\mu)^2})
\sim 
e^{-t\mu^2}\sum_{k, l\geq 0}\frac{(2t\mu)^k}{k!}\rho_{l,k}(t),\quad t\to0^+.
\end{equation}
More precisely,
$$
\tr(e^{-t(|D|-\mu)^2})-e^{-t\mu^2}\sum_{\substack{0\leq k\leq K\\r_{l+1} \leq r_0+\frac{K-k}{2}}}\frac{(2t\mu)^k}{k!}\rho_{l,k}(t)
=
o_(t^{r_0+\frac{K}{2}}), \quad t \to 0^+.
$$

\bigskip

We should also compute the asymptotic expansion of $\tr(e^{-t(D^2-\mu)^2})$. We denote $\mathcal{Y}(s):=\Gamma(s)\zeta_{D^2}(2s)$.
According to \cite[Corollary 3.9]{Eckstein:2019dcb},
we obtain the  asymptotic expansion:
\begin{equation}\label{eq_ht_D4}
\tr(e^{-tD^4})
\sim
\sum_l\Tilde{\rho}_l(t),\quad \Tilde{\rho}_l(t)=\sum_{z\in X_l}\Tilde{a}_z t^{-z}, \quad t\to 0^+,
\end{equation}
with $\Tilde{a}_z=\Res(\mathcal{Y}(s), z)$.
Moreover, we have 
\begin{equation*}
 \tr(D^{2k}e^{-tD^4})
\sim
\sum_l\Tilde{\rho}_{l,k}(t),
\end{equation*}
 where 
\begin{align*}
\Tilde{\rho}_{l,k}(t)
=
&\sum_{\substack{z\in X_l\\(z+k)/2\notin \mathbb{Z}^-}}\frac{\Gamma((z+k)/2)}{2}\Res(\zeta_{D^2}(s), z)t^{-\frac{z+k}{2}}\\ 
&+
\sum_{\substack{z\in X_l\\(z+k)/2\in \mathbb{Z}^-}}\frac{(-1)^{(z+k)/2}}{(-(z+k)/2)!}\zeta_{D^2}(z)t^{-\frac{z+k}{2}}.
\end{align*}

Thus we get  the asymptotic expansion
\begin{equation}\label{eq_ht_D4_mu}
\tr(e^{-t(D^2-\mu)^2})
\sim 
e^{-t\mu^2}\sum_{k, l\geq 0}\frac{(2t\mu)^k}{k!}\Tilde{\rho}_{l,k}(t),\quad t\to0^+,
\end{equation}
so that 
$$
\tr(e^{-t(D^2-\mu)^2})-e^{-t\mu^2}\sum_{\substack{0\leq k\leq K\\r_{l+1} \leq r_0+K-k}}\frac{(2t\mu)^k}{k!}\Tilde{\rho}_{l,k}(t)
=
o (t^{\frac{r_0+K}{2}}), \quad t\to 0^+.
$$

\section{Spectral action basics} \label{Appendix_sa}

In this appendix we shall briefly recall  the spectral action principle, originally formulated by Chamseddine and Connes \cite{Chamseddine:1996zu}. Note that we assume a slightly weaker condition on the spectral triple compared to our assumptions in Appendix B in Eq. \eqref{eq_ht}. Also we use a slightly different notation in this Appendix, in conformity with the original formulation in  \cite{Chamseddine:1996zu}.
Assume $(\mathcal{A}, \mathcal{H}, D)$ is a finitely summable regular spectral triple with simple dimension spectrum, 
The spectral action is defined as 
\begin{equation*}
\textrm{Tr}(f(D/\Lambda)),
\end{equation*}
where $f(x)$ is a non-negative even smooth function which is rapidly decreasing at $\pm \infty$,
and $\Lambda$ is a positive number called mass scale, or cutoff.  Note that  $f(D/\Lambda)$  is  a trace-class operator.
We denote by $\chi(x)=f(\sqrt{x})$, and 
assume that $\chi(x)$ is given as a Laplace transform
\begin{equation*}
\chi(x) = \int_0^\infty e^{-sx}g(s)ds,
\end{equation*}
where $g(s)$ is rapidly decreasing near $0$ and $\infty$.
We also assume that there is a heat trace expansion
\begin{equation*}
\textrm{Tr}\left(e^{-tD^2}\right)
\sim
\sum_{\alpha}a_\alpha t^\alpha, \qquad t\to 0^+,
\end{equation*}
It was proved in \cite{Chamseddine:1996zu} that the spectral action has an asymptotic expansion  for $\Lambda\to \infty$, 
namely,
\begin{equation*}
\textrm{Tr}(\chi(D^2/\Lambda))
\sim
\sum a_\alpha \Lambda^{-\alpha} \int_0^\infty s^\alpha g(s)ds.
\end{equation*}
When $\alpha < 0$,
by the Mellin transform,
\begin{equation*}
s^\alpha
=
\frac{1}{\Gamma(-\alpha)}\int_0^\infty e^{-sx}x^{-\alpha-1}dx.
\end{equation*}
Thus the spectral action coefficient is 
\begin{equation*}
\int_0^\infty s^\alpha g(s)ds
=
\frac{1}{\Gamma(-\alpha)}\int_0^\infty \chi(x)x^{-\alpha-1}dx.
\end{equation*}
When $\alpha=0$,
\begin{equation*}
\int_0^\infty g(s)ds=\chi(0).
\end{equation*}
When $\alpha > 0$,
the spectral action coefficient $\int_0^\infty s^\alpha g(s)ds$ is of order $\Lambda^{-1}$.
Thus we get
\begin{equation*}
\textrm{Tr}(\chi(D^2/\Lambda))
\sim
\sum_{\alpha < 0} a_\alpha \Lambda^{-\alpha} \frac{1}{\Gamma(-\alpha)}\int_0^\infty \chi(x)x^{-\alpha-1}dx + a_0 \chi(0) + \mathcal{O}(\Lambda^{-1}), \quad \Lambda\to \infty.
\end{equation*}

And when $\alpha=n$ is a positive integer, 
since $(\partial_x)^{n}(e^{-sx})=(-1)^ns^ne^{-sx}$,
we have that
\begin{equation*}
\int_0^\infty s^n g(s)ds 
= 
(-1)^n\left(\int_0^\infty (\partial_x)^n(e^{-sx})g(s)ds\right)\Bigg |_{x=0}
=
(-1)^n\chi^{(n)}(0).
\end{equation*}

\black


\begin{thebibliography}{99}

\bibitem{BR97}
O.~Bratteli and D.~W. Robinson.
\newblock {\em Operator algebras and quantum-statistical mechanics. {II}
  (Second edition)}.
\newblock Springer-Verlag, New York, 1997.


\bibitem{bateman1954tables}
Bateman~Manuscript Project, H.~Bateman, A.~Erd{\'e}lyi, and United
  States Office of~Naval~Research.
\newblock {\em Tables of Integral Transforms: Based, in Part, on Notes Left by
  Harry Bateman}.
\newblock Number v. 1 in Bateman manuscript project. McGraw-Hill, 1954.


\bibitem{Chamseddine:1996zu}
A. H. Chamseddine and A.  Connes.
\newblock {The Spectral action principle}.
\newblock {\em Commun. Math. Phys.}, 186:731--750, 1997.

\bibitem{CCS18}
A.~H. Chamseddine, A.~Connes, and W.~D. van Suijlekom.
\newblock Entropy and the spectral action.
\newblock {\em Comm. Math. Phys.} 373 (2020)  457--471.

\bibitem{C94}
A.~Connes.
\newblock {\em Noncommutative Geometry}.
\newblock Academic Press, San Diego, 1994.



\bibitem{Eckstein:2019dcb}
M.~Eckstein and B.~Iochum.
\newblock {\em Spectral Action in Noncommutative Geometry}.
\newblock Springer Briefs in Mathematical Physics. Springer International
  Publishing, 2019.

\bibitem{fegan1985}
H.~D. Fegan and P. Gilkey.
\newblock Invariants of the heat equation.
\newblock {\em Pacific J. Math.}, 117(2):233--254, 1985.

\bibitem{gradshteyn2007}
I.~S. Gradshteyn and I.~M. Ryzhik.
\newblock {\em Table of integrals, series, and products}.
\newblock Elsevier/Academic Press, Amsterdam, seventh edition, 2007.
\newblock Translated from the Russian, Translation edited and with a preface by
  Alan Jeffrey and Daniel Zwillinger. 


\bibitem{GS96}
 G. Grubb and R.T. Seeley.
  \newblock  Zeta and eta functions for Atiyah--Patodi--Singer operators.
  \newblock {\em J. Geom. An.} 6 (1996), 31–77.
  
\bibitem{temme1996special}
N.M. Temme.
\newblock {\em Special Functions: An Introduction to the Classical Functions of
  Mathematical Physics}.
\newblock A Wiley-Interscience Publication. Wiley, 1996.









\end{thebibliography}
\end{document}